\documentclass{llncs}

\usepackage[linesnumbered,ruled,noend]{algorithm2e}
\usepackage{amssymb}
\usepackage{amsmath}
\usepackage{color}
\usepackage{enumerate}
\usepackage{graphicx}
\usepackage{lipsum}
\usepackage{listings}
\usepackage{tikz}
\usepackage{url}
\usepackage{wrapfig}
\usepackage{xspace}

\newenvironment{list1}{\begin{list}{$\bullet$}
		{\topsep 0 pt \parsep 0 pt \partopsep 0 pt \itemsep 0
			pt}}{\end{list}}

\newenvironment{ProcessStep}[1][htb]
{
	\begin{algorithm}[#1]%
	}{\end{algorithm}}

\newcommand{\comments}[1]{}

\newcommand{\G}{\textbf{G}\xspace}
\newcommand{\F}{\textbf{F}\xspace}
\newcommand{\U}{\textbf{U}\xspace}
\newcommand{\x}{\textbf{X}\xspace}
\newcommand{\CU}{\textbf{W}\xspace}
\newcommand{\GXW}{\textsf{GXW}\xspace}
\newcommand{\sig}[1]{{\small\textsf{#1}\xspace}}
\newcommand{\s}[1]{\textsf{{#1}}}
\newcommand{\A}{\mathcal{A}}
\newcommand{\C}{\mathcal{C}}
\newcommand{\Sys}{\mathcal{S}}
\newcommand{\Map}{\textsf{map}}

\title{Structural Synthesis for \GXW Specifications}

\author{}
\institute{}

\author{Chih-Hong Cheng 	\and Yassine Hamza    \and Harald Ruess}

\institute{
	fortiss - An-Institut Technische Universit{\"a}t  M\"{u}nchen\\
	Guerickestr. 25, 80805 Munich, Germany\\
	{\tt \{cheng,ruess\}@fortiss.org}, {\tt yassine.hamza@in.tum.de}
}

\begin{document}

\maketitle
\vspace{-5mm}
 \begin{abstract}	
 We define the $\GXW$ fragment of linear temporal logic (\s{LTL}) as the basis for synthesizing
 embedded control software for safety-critical applications.  
 Since $\GXW$ includes the use of a {\em weak-until} operator we are able 
 to specify a  number of diverse programmable logic control (\s{PLC}) problems, 
 which we have compiled from industrial training sets.
 For $\GXW$ controller specifications, we develop a novel approach for 
 synthesizing a set of synchronously communicating actor-based controllers.
 This synthesis algorithm proceeds by means of recursing over the
 structure of $\GXW$ specifications, and generates a set of dedicated and synchronously communicating sub-controllers according to the formula structure.
 In a subsequent step, \s{2QBF} constraint solving identifies and tries to resolve potential conflicts between individual $\GXW$ specifications. 
 This structural approach to $\GXW$ synthesis supports traceability between requirements and the generated control code as mandated by certification 
 regimes for safety-critical software.
 Synthesis for \GXW specifications is in \s{PSPACE} compared to \s{2EXPTIME}-completeness of full-fledged \s{LTL} synthesis. 
 Indeed our experimental results suggest that \GXW synthesis scales 
 well to industrial-sized control synthesis problems with 20 input and output ports and beyond.
 \end{abstract}

\section{Introduction}\label{sec.introduction}
\vspace{-1mm}
Embedded control software in the manufacturing and processing industries is usually developed using specialized programming languages such as ladder diagrams or other IEC 61131-3 defined languages. 
Programming in these rather low-level languages is not only error-prone but also time- and resource-intensive.
Therefore we are addressing the problem of correct-by-construction and automated generation of embedded control software from high-level requirements, which are expressed in a suitable fragment of linear temporal logic.

Moreover, an explicit correspondence between the high-level requirements and the generated control code is essential, 
since embedded control software
is usually an integral part of safety-critical systems such as  supervisory control and data acquisition (\s{SCADA}) systems for controlling critical machinery or infrastructure. 
In particular current industrial standards for safety-related development such as IEC 61508, DO 178C for avionics, and ISO 26262 for automotive applications mandate 
traceability between the control code and it requirements. 
Controllers generated by state-of-the-art \s{LTL} synthesis algorithms and tools such as generalized reactivity(1) (\s{GR(1)})~\cite{anzu,gr1} 
or bounded LTL synthesis~\cite{acacia12,Ehlers11,ScheweF07a}, however, usually do not explicitly support such traceability requirements.
For example, the \s{GR(1)} synthesis tool Anzu generates circuit descriptions in Verilog from
BDDs~\cite{anzu}\@. 

We are therefore proposing a novel approach for synthesizing structured control software.
In essence, the control code is generated by means of structural recursion on the given \s{LTL} formulas. 
Therefore, the structure of the control code corresponds closely to the syntactic structure of the given requirements, and there is a direct correspondence between
controller components and sub-formulas of the specification.

In a first step towards this goal, we identify a fragment of \s{LTL} for specifying the input-output behavior of typical embedded control components. 
Besides the specification of input assumptions, invariance conditions on outputs, and transition-like reactions of the form $\G (\s{input} \rightarrow \x^i  \s{output})$,
this fragment also contains specifications of reactions of the form
$\G (\s{input} \rightarrow \x^i  (\s{output}\, \CU\,\s{release}  ))$, where $\s{input}$ is an LTL formula containing at most $i$ consecutive \x operators (i.e., an LTL formula whose validity is determined by the next~$i$ input valuations). 
The latter reaction formula states that if there is a temporal input event satisfying 
the constraint $\s{input}$, 
then the $\s{output}$ constraint should hold on output events until there is a $\s{release}$ event (or $\s{output}$ always holds)\@. 
The operator $\G$ is  the universal path quantifier,
$\x^i$ abbreviates $i$ consecutive next-steps,
 $\CU$ denotes the {\emph{weak until}} temporal operator, the constraint $\s{output}$ contains no temporal operator, and the subformula $\s{release}$ may contain certain numbers of consecutive next-steps but no other temporal operators. The resulting fragment of \s{LTL} is called \GXW.

So far we have successfully modelled more than $70$ different embedded control scenarios in $\GXW$\@.  
The main source for this set of benchmarking problems are publicly available collections of industrial training 
materials for PLCs (including CODESYS 3.0 and AC500)~\cite{source1blogspot,source2kaftanABB,source3petry}\@.
The proposed \GXW fragment of \s{LTL} is also similar to established requirements templates for 
specifying embedded control software in the aerospace domain, such as \s{EARS}~\cite{ears}\@. 

Previous work on \s{LTL} synthesis (e.g.,~\cite{acacia12,Ehlers11,ScheweF07a,fmcad06,DBLP:conf/cav/ChengHRS14,anzu,gr1,ratsy,gr1robots,ShieldSynthesis})
usually generates gate-level descriptions for the
synthesized control strategies. 
In contrast, we generate control software in an actor
language with high-level behavioral constructs and synchronous dataflow communication between connected actors.
This choice of generating {\em
structured controllers} is motivated by current practice of programming controllers using, say,  Matlab Simulink~\cite{simulink}, continuous function charts (IEC~61131-3), and Ptolemy II~\cite{ptII}, which also supports synchronous dataflow (\s{SDF}) models~\cite{sdf}\@.  
Notice, however, that the usual notions of \s{LTL} synthesis also apply to 
synthesis for \s{SDF}, since the composition of actors in \s{SDF} may also be viewed as Mealy machines with synchronous cycles~\cite{sdfcomposition}\@.

Synthesis of structured controllers from \GXW specifications proceeds in two subsequent phases.
In the first phase, the procedure recurses on the structure of the given \GXW formulas for generating dedicated actors for monitoring inputs events, for generating corresponding control events, 
and for wiring these actors according
to the structure of the given \GXW formulas. 
In the second phase, appropriate values for unknown parameters 
are synthesized in order to realize the conjunction of all given \GXW specifications.
Here we use satisfiability  checking for quantified Boolean formula (\s{2QBF}) 
for examining if there exists such conflicts between multiple \GXW 
specifications. More precisely, existential variables of generated \s{2QBF} problems capture the remaining design freedom when an output variable is not constrained by any trigger of low-level events.  

We demonstrate that controller synthesis for the \GXW fragment is in \s{PSPACE} as compared to the \s{2EXPTIME}-completeness result of 
full-fledged \s{LTL}~\cite{popl89}\@.  Under some further reasonable syntactic restrictions on the \GXW fragment we show that synthesis is in \s{coNP}\@.


An implementation of our \GXW structural synthesis algorithm and application to our benchmark studies demonstrates a substantial speed-up compared to existing \s{LTL} synthesis tools. 
Moreover, the structure of the generated control code in \s{SDF} follows the structure of the given \GXW specifications, and is more compact and, arguably, also more
readable and understandable than commonly used gate-level representations for synthesized control strategies.

The paper is structured as follows. 
We introduce in Section~\ref{sec.formulation} some basic notation for \s{LTL} synthesis, a definition of the~\GXW  fragment of \s{LTL}\@ and \s{SDF} actor systems together with the 
problem of actor-based \s{LTL} synthesis under~\GXW fragment.
Section~\ref{sec.example} illustrates \GXW and actor-based control for such 
specifications by means of an example. 
Section~\ref{sec.algorithms} includes the main technical contributions and describes algorithmic workflow for generating structured controllers from \GXW, together with soundness and complexity results for~\GXW synthesis.
A summary of our experimental results is provided in Section~\ref{sec.evaluation},
and a comparison of \GXW synthesis with closely related work on \s{LTL} synthesis is included in Section~\ref{sec.related}\@.
The paper closes with concluding remarks in Section~\ref{sec.conclusion}\@.

 \vspace{-2mm}
\section{Problem Formulation}\label{sec.formulation}
\vspace{-2mm}

We present basic concepts and notations of \s{LTL} synthesis, and we
define the \GXW fragment of \s{LTL} together with the problem of 
synthesizing actor-based synchronous dataflow controllers for \GXW.


\vspace{-2mm}
 \subsection {LTL Synthesis}

Given two disjoint sets of Boolean variables $V_{in}$ and $V_{out}$,
the {\em linear temporal logic} (\s{LTL}) formulae over $\mathbf{2}^{V_{in}\cup V_{out}}$ is the smallest set such that
(1)  $v \in\mathbf{2}^{V_{in}\cup V_{out}}$ is an \s{LTL} formula, 
(2)  if $\phi_1, \phi_2$ are \s{LTL}-formulae, then so are $\neg \phi_1$, $\neg \phi_2$, $\phi_1 \vee \phi_2$, $\phi_1 \wedge \phi_2$, $\phi_1 \rightarrow \phi_2$, and 
(3)  if $\phi_1, \phi_2$ are \s{LTL}-formulae, then so are  $\G \phi_1$, $\x \phi_1$, $\phi_1 \U \phi_2$\@. 
Given an $\omega$-word $\sigma$, define $\sigma(i)$ to be the $i$-th element in $\sigma$, and define $\sigma^i$ to be the suffix $\omega$-word of $\sigma$ obtained by truncating $\sigma(0)\ldots\sigma(i-1)$. 
The satisfaction relation $\sigma \vDash \phi$ between an $\omega$-word~$\sigma$ and an LTL formula~$\phi$ is defined in the usual way.
The \emph{weak until} operator, denoted \CU, is similar to the \emph{until} operator but the stop condition is not required to occur;
therefore $\phi_1 \CU \phi_2$ is simply defined as $(\phi_1 \U \phi_2) \vee \G \phi_1$\@.
Also, we use the abbreviation~$\x^i \phi$ to abbreviate $i$ consecutive $\x$ operators before $\phi$\@.

A deterministic \emph{Mealy machine} is a finite automaton $\C = (Q, q_0,  \mathbf{2}^{V_{in}}, \mathbf{2}^{V_{out}}, \delta)$, where $Q$ is set of (Boolean) state variables (thus $\mathbf{2}^{Q}$ is the set of states), $q_0 \in \mathbf{2}^{Q}$ is the initial state,  $\mathbf{2}^{V_{in}}$ and $\mathbf{2}^{V_{out}}$ are sets of all input and output assignments defined by two disjoint sets of variables $V_{in}$ and $V_{out}$\@. 
$\delta :=  \mathbf{2}^{Q} \times \mathbf{2}^{V_{in}} \rightarrow  \mathbf{2}^{V_{out}} \times \mathbf{2}^{Q}$ is the transition function that takes
     (1) a state $q\in \mathbf{2}^{Q}$  and (2) input assignment $v_{in} \in \mathbf{2}^{V_{in}}$, and returns (1) an output assignment $v_{out} \in \mathbf{2}^{V_{out}}$ and 
     (2) the successor state $q' \in \mathbf{2}^{Q}$. 
Let $\delta_{out}$ and $\delta_{s}$ be the projection of $\delta$ which considers only output assignments and only successor states. 
Given a sequence $a_0  \ldots a_{k}$ where $\forall i =0\ldots k$, $a_i \in \mathbf{2}^{V_{in}}$, let $\delta_{s}^{k} (q_0, a_0  \ldots a_{k})$ 
abbreviate the output state derived by executing $a_0  \ldots a_{k}$ as an input sequence on the Mealy machine.

Given a set of input and output Boolean variables $V_{in}$ and $V_{out}$, together with an LTL formula $\phi$ on $V_{in}$ and $V_{out}$ 
the \emph{LTL synthesis problem} asks the existence of a controller as a deterministic Mealy machine $\C_{\phi}$ such that, for every input sequence $a= a_0a_1a_2\ldots$, where $a_i \in \mathbf{2}^{V_{in}}$:
	(1) given the prefix $a_0$ produce $b_0 = \delta_{out}(q_0, a_0)$,
	(2) given the prefix $a_0 a_1$ produce $b_1 = \delta_{out}(\delta_{s}(q_0, a_0), a_1)$,
	(3) given the prefix $a_0  \ldots a_{k} a_{k+1}$, produce $b_{k+1} = \delta_{out}(\delta_{s}^{k} (q_0, a_0  \ldots a_{k}), a_{k+1})$, and
	(4) the produced output sequence $b= b_0 b_1\ldots$ ensures that the word $\sigma = \sigma_1 \sigma_2 \ldots$, where $\sigma_i =  a_i b_i \in \mathbf{2}^{V_{in}\cup V_{out}}$, $\sigma \vDash \phi$.

\begin{table}[t]
	\begin{minipage}[b]{0.49\linewidth}\centering
		\centering
		\begin{tabular}{l|l|l}
			ID & Meaning & Pattern \\\hline 
			P1&	Initial-Until		& $\varrho_{out}\,\CU\,\phi^{i}_{in}$  \\
			P2&	Trigger-Until	& $\G (\phi^i_{in} \rightarrow \x^i  (\varrho_{out}\,\CU\, (\varphi^{j}_{in} \vee \rho^0_{out}))  )$ \\
			P3&	If-Then		& $\G (\phi^i_{in} \rightarrow \x^i \varrho_{out})$  \\
			P4&	Iff	& $\G (\phi^i_{in} \leftrightarrow \x^i \varrho_{out})$  \\
			P5&     Invariance  &  $\G (\phi^0_{out})$ \\
			P6&      Assumption &  $\G (\phi^0_{in})$ \\\hline	
		\end{tabular}
		    \vspace{2mm}
			\caption{Patterns defined in \GXW specifications}
			\label{table.specification.pattern}
	\end{minipage}
	\hspace{0.5cm}
	\begin{minipage}[b]{0.45\linewidth}
		\centering
		\begin{tabular}{c|l}
			Pattern ID  & High-level Control Specification \\\hline
			P1& 	$\s{output} \,\CU\, \s{input}$	 \\
			P2&	$\G (\s{input} \rightarrow (\s{output} \,\CU\, \s{release})) $	 \\
			P3&	$\G (\s{input} \rightarrow \s{output})$	 \\\hline
		\end{tabular}
		    \vspace{2mm}
			\caption{Specification patterns and corresponding skeleton specification.}
			\label{table.specification.skeleton}		
	\end{minipage}
\end{table}

\vspace{-2mm}
 \subsection{\GXW Synthesis}

We formally define the \GXW fragment of \s{LTL}\@. 
Let $\phi^i$, $\varphi^i$, $\psi^i$ be \s{LTL} formulae over input variables $V_{in}$ and output variables $V_{out}$\@, where all formulas are (without loss of generality) assumed to be in disjunctive normal form (DNF), and each literal is of form $\x^j\,v$ or $\neg \x^j\, v$ 
with  $0\leq j\leq i$ and $v \in V_{in}\cup V_{out}$\@. 
Clauses in DNF are also called \emph{clause formulae}\@. 
Moreover, a formula $\phi^i_{in}$ is restricted to contain only 
input variables in $V_{in}$,
and similarly, $\phi^i_{out}$ contains only output variables in $V_{out}$\@.
Finally,  $\varrho_{out}$ denotes either $v_{out}$ or  $\neg v_{out}$, where $v_{out}$ is an output variable.

	For given input variables  $V_{in}$ and output variables $V_{out}$, a
	$\GXW$ {\em formula} is an \s{LTL} formula of one of the forms (P1)-(P6) as specified in Table~\ref{table.specification.pattern}\@.
%
For example, \GXW formulas of the form (P2) stop locking $\varrho_{out}$ as soon as $(\varphi^{j}_{in} \vee\rho^0_{out})$ holds.
{\em \GXW specifications} 
are of the form
\begin{equation} \label{eq:synthesized.spec}
\varrho \rightarrow \bigwedge_{m=1\ldots k} \eta_m\mbox{~,}
\end{equation}
where $\varrho$ matches the \GXW pattern~(P6), and $\eta_m$ matches one 
of the patterns~(P1) through~(P5) in Table~\ref{table.specification.pattern}\@.
Furthermore, the notation ``.'' is used for projecting subformulas from $\eta_m$, 
when it satisfies a given type.  
For example, assuming that sub-specification $\eta_m$ is of pattern~P3, i.e., it matches  $\G (\phi^i_{in} \rightarrow \x^i \varrho_{out})$, $\eta_m.\varrho_{out}$ specifies the matching subformula for $\varrho_{out}$\@. 
Notice also that \GXW specifications, despite including the~\CU operator, have the {\em finite model property}, since the smallest number of unrolling steps for disproving the existence of an implementation is linear with respect to the structure of the given formula (cmp. Section~\ref{sub.sec.properties})\@. 


Instead  of directly synthesizing a Mealy machine as in standard LTL synthesis, we are considering here the generation of
{\em actor-based controllers} using the computational model of {\em synchronous dataflow} (\s{SDF}) without feedback loops.  
An {\em actor-based controller} is a tuple $\Sys= (\mathcal{V}_{in}, \mathcal{V}_{out}, Act, \tau)$, where
	$\mathcal{V}_{in}$ and $\mathcal{V}_{out}$ are disjoint sets of external input and output ports.
	Each port is a variable which may be assigned a Boolean value
	or \s{undefined} if no such value is available at the port\@.
	In addition, actors $\A \in Act$ may be associated with internal input ports $U_{in}$ and
	output ports $U_{out}$ (all named apart), which are also three-valued\@.
    The projection $\A.\s{u}$ denotes the port $\s{u}$ of $\A$\@.  
	An actor  $\A \in Act$ defines Mealy machine $\C$ whose input and output assignments are based on $\mathbf{2}^{U_{in}}$ and $\mathbf{2}^{U_{out}}$, i.e., the output update function of $\C$ sets each output port to \s{true} or \s{false},  when each input port has value in \{\s{true}, \s{false}\}\@. 
	 Lastly, $\A^{(i)}$ denotes a \emph{copy} of $\A$ which is indexed by~$i$\@. 
	 
	Let $Act.U_{in}$ and $Act.U_{out}$ be the set of all internal input and output ports for $Act$\@. The wiring $\tau \subseteq (\mathcal{V}_{in}\cup Act.U_{in}) \times (\mathcal{V}_{out}\cup Act.U_{out})$ connects one (external, internal) input port to one or more (external, internal) output ports. For convenience, denote the wiring from  port $\s{out}$ of $\A_1$ to port $\s{in}$ of $\A_2$ as ($\A_1.\s{out} \dashrightarrow  \A_2.\s{in}$).	
	All ports are supposed to be connected, and every internal input port and every external output port is only connected to one wire 
	(thus a port does not receive data from two different sources). 
	Also, we do not consider actor systems with feedback loops here (therefore no cycles such as the one in Figure~\ref{fig:Actor.Composition}(c)), since  systems without feedback loops can be statically scheduled~\cite{sdfscheduling}\@. 

Evaluation cycles are triggered externally under the semantics of synchronous dataflow.
In each such cycle, the data received at the external input ports is processed and corresponding values are transferred to external output ports. 
Notice also that the composition of actors under \s{SDF} acts cycle-wise 
as a Mealy machine~\cite{sdfcomposition}\@. 
We illustrate the {\em operational semantics} of actor-based systems under  \s{SDF}
by means of the example in Figure~\ref{fig:Actor.Composition}(a), with
input ports {\s{in1}}, {\s{in2}}, output port {\s{out}},
and actors $f_1$, $f_2$, $f_3$, $f_4$ (see also Figure~\ref{fig:Actor.Composition}(b))\footnote{The formal operational semantics, as it is standardized notation from \s{SDF}, is relegated to the appendix.}. 
Now, assume that in the first cycle, the input ports {\s{in1}} and {\s{in2}} receive the value (\s{false}, \s{true}) and in the second cycle the value (\s{false}, \s{true})\@. The \s{false} value in {\s{in1}} is copied to $f_1$.{\s{i}}. As $f_1$ is initially at state where $v=\s{false}$, it creates the output value \s{true} (places it to $f_1$.\s{o}) and changes its internal state to $v=\s{true}$. The value \s{true} from $f_1$.\s{o} is then transferred to $f_4.\s{i}_1$ and $f_2.\s{i}_1$. However, at this stage one cannot evaluate $f_2$ or $f_4$, as the $\s{i}_2$ port is not yet filled with a value. $f_3$ receives the value from ${\s{in2}}$ and produces $f_3.\s{o}$ to \s{false}. 
Continuing this process, at the end of first cycle ${\s{out}}$ is set to \s{true}, while in the second cycle,  ${\s{out}}$ is set to \s{false}. 

\begin{figure}[t]
	\centering
	\begin{minipage}{0.55\textwidth}
	\centering
	\includegraphics[width=\columnwidth]{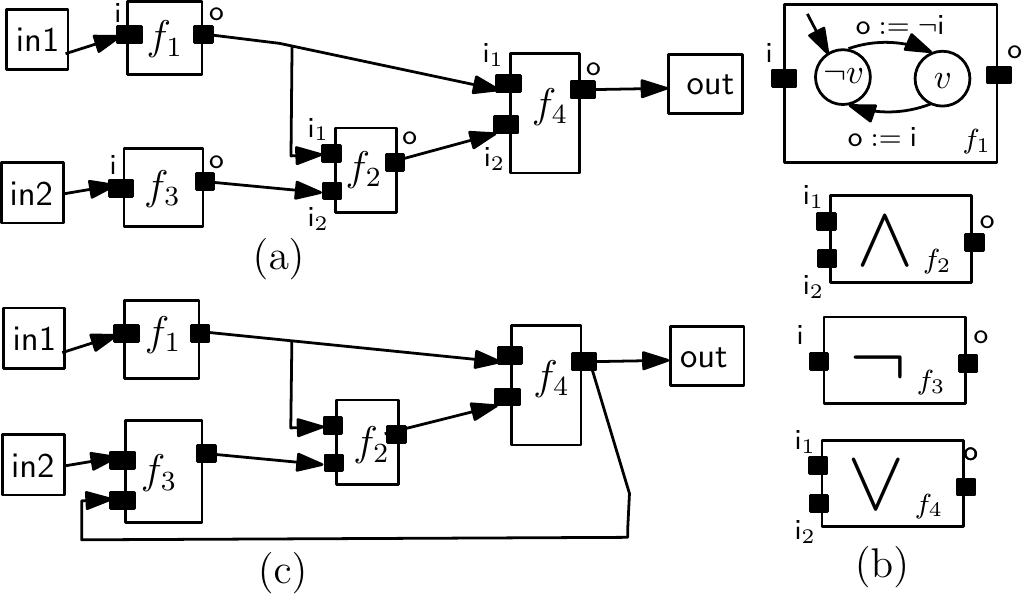}	
	\caption{An actor system allowing functional composition and corresponding actors $f_1$~(a)(b), feedback loops such as~(c) are not considered here.}	
	\label{fig:Actor.Composition}
	\end{minipage}\hfill
	\begin{minipage}{0.45\textwidth}
	\includegraphics[width=\columnwidth]{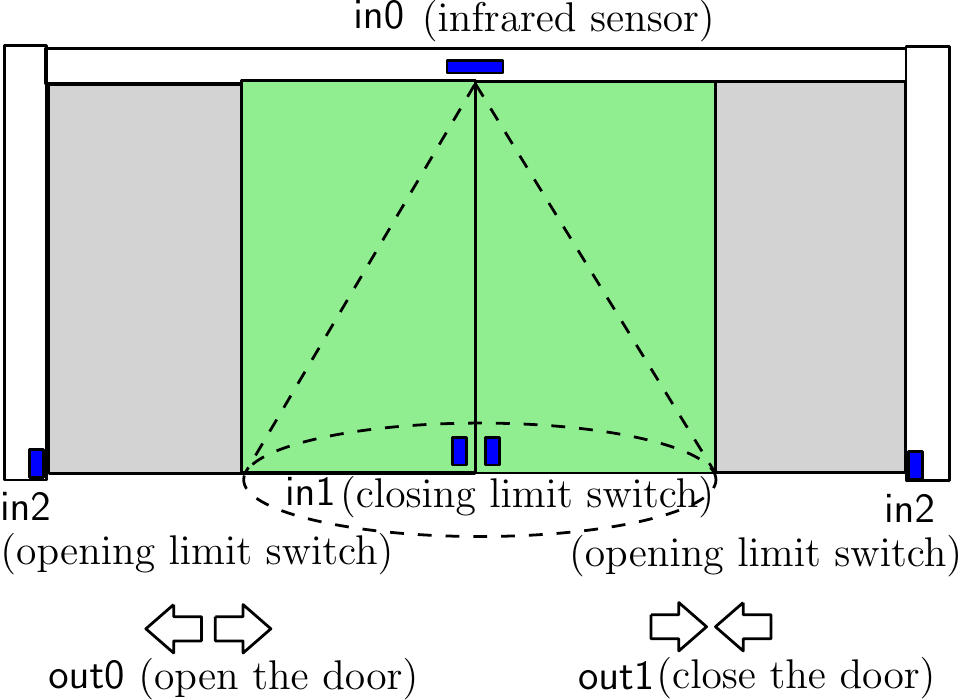}	
	\caption{Control of automatic door switch.}	
	\label{fig:Automatic.Door}
	\end{minipage}
\end{figure}

As we do not consider feedback loops between actors in $Act$, from input read to output write, one can, using the enumeration method as exemplified above, create a static linear list $\Xi$ of size $|Act|+|\tau|$, where each element $\xi_{ind} \in \Xi$ is either in $Act$ or in $\tau$, for specifying the linear order (from the partial order) how data is transferred between wires and actors. 
Such a total order $\Xi$ is also called an \emph{evaluation ordering} of the actor
system $\Sys$\@. 

%





One may wrap any Mealy machine $\C$ as an actor $\A(\C)$ by
simply 
creating corresponding ports in $\A(\C)$ and by setting
the underlying Mealy machine of $\A(\C)$ to $\C$\@. 
Therefore, actor-based controllers may be synthesized for a given \s{LTL} specification $\phi$ by first  synthesizing a Mealy machine~$\C$ realizing $\phi$, followed by the wrapping~$\C$ as~$\A(\C)$, creating external \s{I/O} ports, and connecting external \s{I/O} ports with $\A(\C)$\@. 

Given a \GXW specification $\phi$ over the input variables $V_{in}$ and output variables $V_{out}$, the problem of~\GXW~{\em synthesis} is to generate an 
actor-based SDF controller $\Sys$ realizing $\phi$.
As one can always synthesize a Mealy machine followed by wrapping it to an actor-based controller, \GXW~synthesis has the same complexity for Mealy machine and for actor-based controllers.



\vspace{-2mm}
\section{Example}\label{sec.example}
\vspace{-2mm}



We exemplify the use of $\GXW$ specifications and actor-based synthesis for these kinds of specification by means of an automatic sliding door~\cite{example}, 
which is visualized in Figure~\ref{fig:Automatic.Door}\@.
Inputs and outputs are as follows:
	 $\s{in0}$ is $\s{true}$ when someone enters the sensing field;
	 $\s{in1}$  denotes a closing limit switch - it is $\s{true}$ when two doors touch each other;
	 $\s{in2}$ denotes an opening limit switch - it is $\s{true}$ when the door reaches the end;
	 $\s{out0}$ denotes the opening motor - when it is set to $\s{true}$ the motor rotates clockwise, thereby triggering the door opening action; and	
	 $\s{out1}$ denotes closing motor - when it is set to $\s{true}$ the motor rotates counter-clockwise, thereby triggering the door closing action.
	 Finally, the triggering of a timer  $\s{t0}$ is modeled by means a (controllable) output variable $\s{t0start}$ and the expiration of a timer is modeled using an (uncontrollable) input variable $\s{t0expire}$\@. 

Before stating the formal \GXW specification for the example we introduce some mnemonics.
\comments{
\begin{eqnarray*}
	\s{entering}^1 &:=& \neg \s{in0} \wedge \x\,\s{in0}  \\
	\s{expired}^1 &:=& \neg \s{t0expire} \wedge (\x \,\s{t0expire}) \\
	\s{lim\_reached}^1 &:=& \neg \s{in2} \wedge \x \,\s{in2} \\
	\s{closing\_stopped} &:=& \s{in1} \vee \s{in0} \vee \s{out0}
\end{eqnarray*}
}

 \vspace{1mm}
 \begin{minipage}[t]{.45\textwidth}
	\raggedright
 \begin{itemize}
	\item $\s{entering}^1$ := $\neg \s{in0} \wedge \x\,\s{in0} $ 
	\item $\s{expired}^1$ := $\neg \s{t0expire} \wedge (\x \s{t0expire})$
\end{itemize}
	
 \end{minipage}
 \hfill
 \noindent
\begin{minipage}[t]{.45\textwidth}
	\raggedright	
\begin{itemize}
	\item $\s{lim\_reached}^1$ := $\neg \s{in2} \wedge \x \s{in2}$
	\item $\s{closing\_stopped}$ := $\s{in1} \vee \s{in0} \vee \s{out0}$ 
 \end{itemize}
 \end{minipage}
 \vspace{1mm}

The superscripts denote the maximum number of consecutive next-steps. Now the automatic sliding door controller is formalized in $\GXW$ as follows. 

\begin{figure}[!t]
	\centering
	\includegraphics[width=0.95\columnwidth]{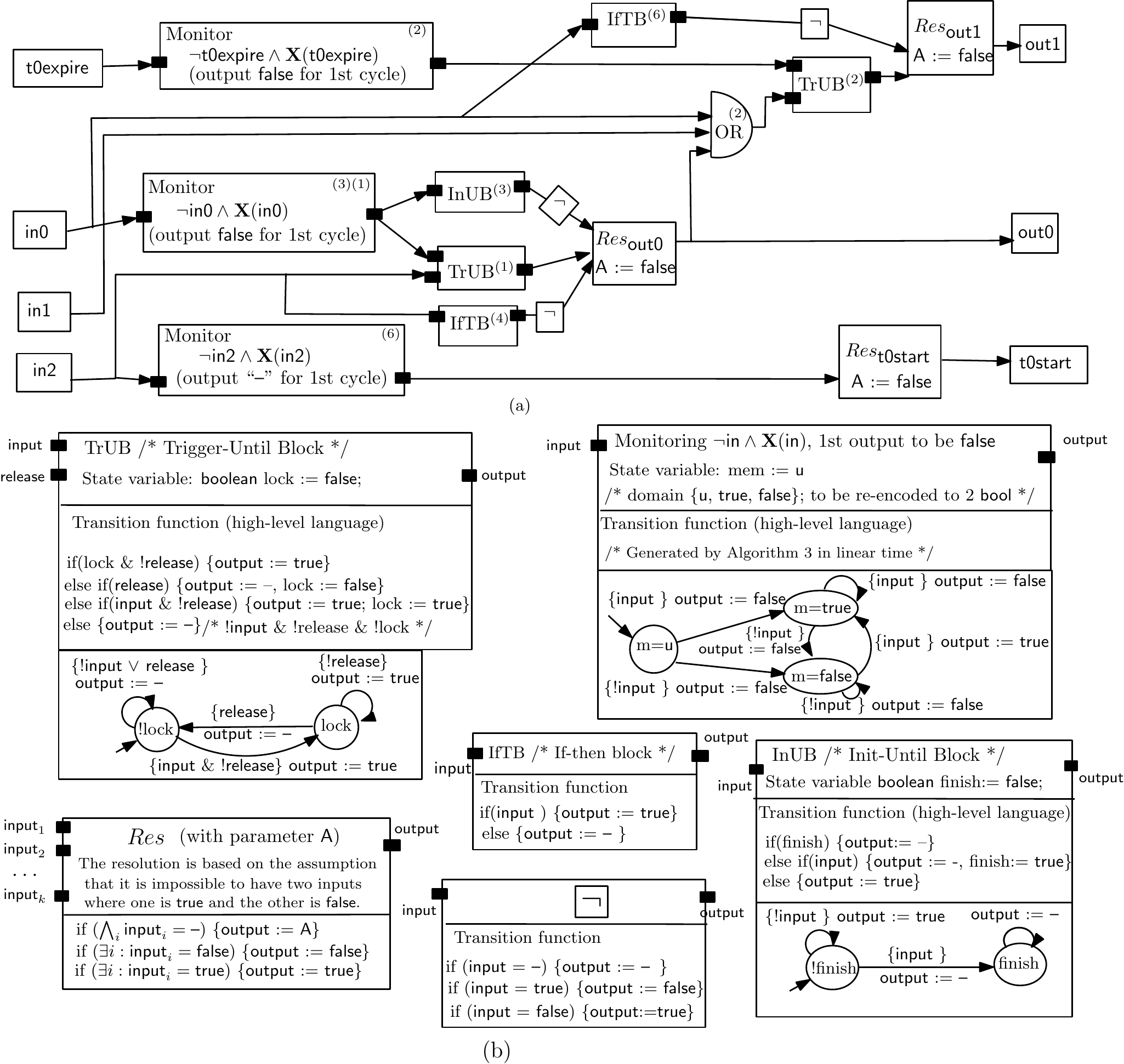}
			\caption{Actor-based controller realizing automatic sliding door.}	
			\label{fig:Resolution}	
	\vspace{1mm}
\end{figure}

\comments{
\begin{eqnarray*}
\mbox{(S1)} && \G (\s{entering}^1 \rightarrow \x (\s{out0}\,\CU\,\s{in2})) \\
\mbox{(S2)} && \G(\s{expired}^1 \rightarrow \x^1 (\s{out1}\,\CU\, \s{closing\_stopped})) \\
\mbox{(S3)} &&\neg\s{out0}\,\CU\,\s{entering}^1 \\
\mbox{(S4)} && \G (\s{in2} \rightarrow \neg \s{out0}) \\
\mbox{(S5)} && \G (\s{lim\_reached}^1 \leftrightarrow  \x (\s{t0start}) \\   \mbox{(S6)} && \G (\s{in0} \rightarrow \neg \s{out1}) \\
\mbox{(S7)} &&\G (\neg(\s{out0} \wedge \s{out1}))
\end{eqnarray*}
}
 \vspace{1mm}
 \begin{minipage}[t]{.5\textwidth}
 	\raggedright
 	
 	\begin{enumerate}[{S}1:]
 		\item 	$\G (\s{entering}^1 \rightarrow \x (\s{out0}\,\CU\,\s{in2}))$ 
 		\item 	$\G(\s{expired}^1 \rightarrow \x (\s{out1}\,\CU\, 
 		\s{closing\_stopped}))$
 		\item 	$\neg\s{out0}\,\CU\,\s{entering}^1$
 		\item 	$\G (\s{in2} \rightarrow \neg \s{out0})$
 	\end{enumerate}
 	
 \end{minipage}
 \hfill
 \noindent
 \begin{minipage}[t]{.35\textwidth}
 	\raggedright
 	
 	\begin{enumerate}[{S}1:]
 		\setcounter{enumi}{4}
 		\item 	$\G (\s{lim\_reached}^1 \leftrightarrow  \x (\s{t0start}))$
 		\item 	$\G (\s{in0} \rightarrow \neg \s{out1})$
 		\item 	$\G (\neg(\s{out0} \wedge \s{out1}))$	
 	\end{enumerate}
 	
 \end{minipage}
 \vspace{1mm}

In particular, formula (S1) expresses the requirement that the opening of the door should continue
($\s{out0}=\s{true}$) until the limit is reached ($\s{in2}$), and formulas (S3) and (S7) 
specify the expected initial behavior of the automatic sliding door. 
The \GXW specifications for the sliding door example are classified as follows: formulas (S1), (S2) are of type (P2), (S3) is of type (P1), (S4), (S6) is of type (P3), (S5) is of type (P4), and (S7) of type (P5) according to Table~\ref{table.specification.pattern}\@. 

Figure~\ref{fig:Resolution} visualizes an actor-based automatic sliding door controller which realizes the \GXW specification (S1)-(S7)\@.
It is constructed from a small number of building blocks, which are also described
in Figure~\ref{fig:Resolution}\@. Monitor actors, for example, are used for monitoring when the \s{entering}, \s{expired}, and \s{lim\_reached} constraints are fulfilled, the \s{OR} actor is introduced because of the \s{closing\_stopped} release condition in specification (S1), and the two copies of the {\em trigger-until} actors are introduced because of the (P2) shape of the specifications (S1) and (S2)\@.  The input and output ports of the 
{\em trigger-until} actor are in accordance with the namings for (P2) in Table~\ref{table.specification.skeleton}\@.  
{\em Resolution actors} are used for resolving potential conflicts between individual \GXW formulas in a specification.
These actors are parameterized with respect to a Boolean \s{A}, which is the output of the resolution actor in case all inputs of this actor may be in $\{\s{true},\s{false}\}$ (this set is denoted by the shorthand ``$\s{\textendash}$'' in Figure~\ref{fig:Resolution})\@.  
The presented algorithm sets up a \s{2QBF} problem for synthesizing possible values for these parameters.
Because of the constraint (S7) on possible outputs \s{out0} and \s{out1}, the parameter \s{A} for the resolution actor for output \s{out0}, for example, needs to be set to \s{A:=false}\@. 
Figure~\ref{fig:Resolution} also includes the operational behavior of selected actors in terms of high-level transitions and/or Mealy machines. The internal state and behavior for monitor actors, however, is synthesized, in linear time, from a given \GXW constraint on inputs (see Section \ref{sec.algorithms})\@.

Finally, the structural correspondence of the actor-based controller in Figure~\ref{fig:Resolution} with the given \GXW specification of the sliding door
example is being made explicit by superscripting actors with index $(i)$ whenever the actor has been introduced due to the $i$-th specification.

\vspace{-2mm}
\section{Structural Synthesis}\label{sec.algorithms}
\vspace{-2mm}


We now describe the algorithmic details for generating structured controllers from the \GXW 
specifications of the form $\varrho \rightarrow \bigwedge_{m=1\ldots k} \eta_m$\@.
The automated sliding door is used as running example for illustrating the result of each step.

First, our algorithm prepares I/O ports, iterates through every formula~$\eta_m$ for creating high-level controllers (Step~\ref{step.1}) based 
on the appropriate \GXW pattern. 
For specifications of types~P1 to~P3, Table~\ref{table.specification.skeleton} lists the corresponding LTL specification (as high-level control objective), where $\s{input}$ and $\s{release}$ are input Boolean variables, $\s{output}$ is an output Boolean variable.

Then, for each \GXW formula, the algorithm constructs actors and wirings for monitoring low-level events by mimicking the DNF formula structure (Steps~\ref{step.2.p123} and~\ref{step.2.p4})\@.
On the structural level of clause formulas in DNF, the algorithm 
constructs corresponding  controllers in linear time (Algorithm~\ref{algo.synthesis.monitor})\@.  
Finally, the algorithm applies 2QBF satisfiability checking (and synthesis of parameters for resolution actors) for guaranteeing nonexistence of 
potential conflicts between different formulas in the \GXW specifications (Step~\ref{step.3})\@.



\vspace{-2mm}

\subsection{High-level Control Specifications and Resolution Actors} 
\vspace{-2mm}

\begin{ProcessStep}[t]
	\begin{footnotesize}
		\SetKwInOut{Input}{Input}
		\SetKwInOut{Output}{Output}
		\Input{LTL specification $\phi = \varrho \rightarrow \bigwedge_{m=1\ldots k} \eta_m$, input variables $V_{in}$, output variables $V_{out}$}
		\Output{Actor-based (partial) controller implementation $\Sys= (\mathcal{V}_{in}, \mathcal{V}_{out}, Act, \tau$), $\Map_{out}$}
		
		\textbf{let} $\Map_{pattern}$ := \{P1 $\mapsto$ InUB, P2 $\mapsto$ TrUB, P3 $\mapsto$ IfTB\}
		
			$\mathcal{V}_{in} := \{\;\boxed{v_{in}}\;|\;v_{in} \in V_{in}\}$	
		\\
			$\mathcal{V}_{out} := \{\;\boxed{v_{out}}\;|\;v_{out} \in V_{out}\}$
		
		
		\ForEach{$\eta_m$, $m=1\ldots k$}{			
			\If{ $(p := \emph{\textsf{DetectPattern}}(\eta_m)) \in \{P1, P2, P3\}$}{				
				Create actor $\A^{(m)}$ from  $\A:=\Map_{pattern}$.get($p$), and add  to $\Sys$\;				
			} \lElseIf{$(p := \emph{\textsf{DetectPattern}}(\eta_m)) \not\in \{P4, P5, P6\}$}{
			\Return \textbf{error}
		}						
	}	
	
	\textbf{let} $\Map_{out}$ := \textsf{NewEmptyMap}()\;
	
	\lForEach{$v_{out} \in V_{out}$}{
		$\Map_{out}$.put($v_{out}$, \textsf{NewEmptyList}());
	}
	
	\ForEach{$\eta_m$, $m=1\ldots k$}{		
		$\Map_{out}$.get($v_{out}$).add($m$), where 		$v_{out}$ is the output variable used in $\tau_{m}.\varrho_{out}$;

	} 
	
	\lForEach{$v_{out} \in V_{out}$}{
		Add	actor $Res_{v_{out}}$ := \textsf{CreateResActor}($\Map_{out}$.get($v_{out}$).size())  to $Act$			
	}
	
	\ForEach{$\eta_m$, $m=1\ldots k$}{		
		\textbf{let}	$v_{out}$ be the variable used in $\tau_{m}.\varrho_{out}$,   $\s{ind}$ := 	$\Map_{out}$.get($v_{out}$).indexOf($m$)\;	
		
		\If	(\tcp*[h]{negation is used in literal}){$\neg v_{out}$ equals $\tau_{m}.\varrho_{out}$}{
			Create a negation actor $\boxed{\neg}^{(m)}$ and add it to $Act$\;	
			$\tau := \tau \cup \{(\A^{(m)}.\s{output} \dashrightarrow  \boxed{\neg}^{(m)}.\s{input}), (\boxed{\neg}^{(m)}.\s{output} \dashrightarrow Res_{v_{out}}.{\s{input}_{\s{ind}}})\}$\;
		}\lElse{
		$\tau := \tau \cup \{(\A^{(m)}.\s{output} \dashrightarrow  Res_{v_{out}}.{\s{input}_{\s{ind}}})\}$
	}

} 	

\lForEach{$v_{out} \in V_{out}$}{		
	$\tau := \tau \cup \{(Res_{v_{out}}.\s{output} \dashrightarrow \boxed{v_{out}})\}$
}

\caption{Prepare external I/O ports, initiate high-level controller and resolution controllers}
\label{step.1}
\end{footnotesize}
\end{ProcessStep}

The initial structural recursion over \GXW formulas is 
described in Step~\ref{step.1}\@.

\vspace{-2mm}
\paragraph{Step 1.1 - Controller for high-level control objectives.} 

Line~1 associates the three high-level controller actors InUB, TrUB, IfTB
with their corresponding pattern identifier.
Implementations for the actors InUB, TrUB, IfTB are listed in  Figure~\ref{fig:Resolution}(b)\@. 
For example, the actor IfTB is used for realizing $\G (\s{input} \rightarrow \s{output})$ in Table~\ref{table.specification.skeleton}\@.  
When \s{input} equals $\s{false}$, the output produced by this actor equals ``$\s{\textendash}$''. This symbol is used as syntactic sugar for the
set $\{\s{true}, \s{false}\}$\@. 
Therefore the output is unconstrained, that is, it is feasible for \s{output} to 
be either \s{true} or \s{false}. The value ``$\s{\textendash}$'' is 
transferred in the dataflow, thereby allowing the delay of decisions when considering multiple specifications influencing the same output variable.

\vspace{-2mm}
\paragraph{Step 1.2 - External I/O ports.} Line~2 and~3 are producing external input and output port for each input variable $v_{in}\in V_{in}$ and output variable $v_{out} \in V_{out}$. 

\vspace{-2mm}
\paragraph{Step 1.3 - High-level control controller instantiation.}
 Lines~4 and~5 iterate through each specification $\eta_m$ for finding the corresponding pattern (using \textsf{DetectPattern}).
 Based on the corresponding type, line~6 creates a high-level controller by copying the content stored in the map. If there exists a specification which does not match one of the patterns, immediately reject (line~7). 
 Notice that pattern~P4 is handled separately in Step~\ref{step.2.p4}\@. 
For the door example, the controller in  Figure~\ref{fig:Resolution}(a) contains the two copies $TrUB^{(1)}$ and $TrUB^{(2)}$ of the trigger-until actor $TrUB$\@;
the subscripts of these copies are tracing the indices of the originating formulas (S1) and (S2)\@.

\vspace{-2mm}
\paragraph{Step 1.4 - Resolution Actors.} This step is to consider all sub-specifications that influence the same output variable $v_{out}$. Line~9 to 11 adds, for each specification $\eta_m$ using $v_{out}$, its index $m$ maintained by $\Map_{out}$.get($v_{out}$).  E.g., for the door example,  specifications S1, S3 and S4 all output $\s{out0}$. Therefore after executing line 10 and~11, we have $\Map_{out}$.get($\s{out0}$)=$\{1, 3, 4\}$, meaning that for variable $\s{out0}$, the value is influenced by S1, S3 and S4. 

For each output variable $v_{out}$, line~12 creates one \emph{Resolution Actor} $Res_{v_{out}}$ which contains one parameter equaling the number of specifications using $v_{out}$ in $\varrho_{out}$. 
Here we make it a simple memoryless controller as shown in Figure~\ref{fig:Resolution}(b) - $Res_{v_{out}}$ outputs $\s{true}$ when one of its inputs is \s{true}, outputs $\s{false}$ when one of its inputs is \s{false}, and outputs $\s{A}$ (which is currently an unknown value to be synthesized later) when all inputs are ``$\s{\textendash}$''.
The number of input pins is decided by calling the map.  
E.g., for $Res_{\s{out0}}$ in Figure~\ref{fig:Resolution}(a), three inputs are needed because  $\Map_{out}$.get($\s{out0}$).size()=3\@.
The output of the high-level controller $A^{(m)}$ is connected to the input of $Res_{v_{out}}$. When negation is needed due to the negation symbol in  $\varrho_{out}$ (line 15), one introduces a negation actor~$\boxed{\neg}$ which negates $A^{(m)}.\s{output}$ when $A^{(m)}.\s{input}$ is \s{true} or \s{false} (line 16,17). 
To ensure that connections are wired appropriately, $\Map_{out}$ is used such that the number ``$\s{ind}$''  records the precise input port of the $Res_{out}$ (line 14).  
Consider again the door example. Due to the maintained list $\{1, 3, 4\}$, $TrUB^{(1)}.\s{output}$ is connected to  $Res_{v_{out}}.\s{input}_1$, i.e., the first input pin of $Res_{v_{out}}$. 
Also, as $\neg \s{out0}$ is used in S3 and S4, the wiring from $\emph{InUB}^{(3)}$ and $\emph{IfTB}^{(4)}$ to $Res_{\s{out0}}$ in Figure~\ref{fig:Resolution}(a) has a negation actor in between.

Lastly, line 19 connects the output port of a resolution actor to the corresponding external output port. If $Res_{v_{out}}$ receives simultaneously $\s{true}$ and $\s{false}$  from two of its input ports, then $Res_{v_{out}}.\s{output}$ needs to be simultaneously $\s{true}$ and $\s{false}$. These kinds of situations are causing unrealizability of \GXW specification, and Step~\ref{step.3} is used for detecting these kinds of inconsistencies\@.  



\begin{ProcessStep}[!t]

\begin{footnotesize}

\SetKwInOut{Input}{Input}
\SetKwInOut{Output}{Output}
\Input{$\phi = \varrho \rightarrow \bigwedge_{m=1\ldots k} \eta_m$, $V_{in}$, $V_{out}$,
	$\Sys= (\mathcal{V}_{in}, \mathcal{V}_{out}, Act, \tau$) from Step~\ref{step.1}.}
\Output{Actor-based (partial) controller $\Sys= (V_{in}, V_{out}, Act, \tau$) by adding more elements}

\ForEach{$\eta_m$, $m=1\ldots k$}{		
	$p$ := \textsf{DetectPattern}($\eta_m$)\;
	
	\If{$p \in \{P1, P2, P3\}$ }{
		Add an OR-gate actor $\s{OR}_{\phi^{i}_{in}}$ with $\s{size}(\eta_m.\phi^{i}_{in})$ inputs to $Act$\;
		
		\ForEach{clause formula $\chi^i_{in}$ from DNF of $\eta_m.\phi^{i}_{in}$}{
			Add $\A(\C)$ to $Act$, where	$\C$ := \textsf{Syn}($\G (\chi^i_{in} \leftrightarrow \x^{i} \sig{out}) \wedge \bigwedge^{i-1}_{z=0}\x^z \neg \sig{out}$, $\s{In}(\chi^i_{in})$, \{\s{out}\})\;

			\lForEach{$v_{in} \in \s{In}(\chi^i_{in})$}{
				$\tau := \tau \cup \{(\boxed{v_{in}} \dashrightarrow  \A(\C).v_{in})\}$
			}

			$\tau := \tau \cup \{(\A(\C).\s{out} \dashrightarrow  \s{OR}_{\phi^{i}_{in}}.\s{in}_{\s{Index}(\chi^{i}_{in}, \phi^{i}_{in})})\}$\;
			
		}
		$\tau := \tau \cup \{(\s{OR}_{\phi^{i}_{in}}.\s{out} \dashrightarrow  \A^{(m)}.\s{input})\}$\;		
		
	}

		\If{$p \in \{P2\}$ 	}{
			Add an OR-gate actor $\s{OR}_{\varphi^{j}_{in} \vee \rho^0_{out}}$ with $\s{size}(\eta_m.(\varphi^{j}_{in} \vee \rho^0_{out}))$ inputs  to $Act$\;
			\ForEach{clause formula $\chi^h_{in}$ from DNF  $\eta_m.\varphi^{j}_{in}$}{
				Add $\A(\C)$ to $Act$, where $\C$ := \textsf{Syn}($\G (\chi^h_{in} \leftrightarrow \x^{h} \sig{out}) \wedge \bigwedge^{h-1}_{z=0} \x^z \neg \sig{out}$, $\s{In}(\chi^h_{in})$, \{\s{out}\})\;
				\lForEach{$v_{in} \in \s{In}(\chi^i_{in})$}{
					$\tau := \tau \cup \{(\boxed{v_{in}} \dashrightarrow  \A(\C).v_{in})\}$
				}
				\lIf{$h = 0$}{
					Add ($\A(\C).\s{out} \dashrightarrow  \s{OR}_{\varphi^{j}_{in}\vee \rho^0_{out}}.\s{in}_{\s{Index}(\chi^{z}_{in}, \varphi^{j}_{in}\vee \rho^0_{out})}$) to $\tau$
				} \uElse {
				Add $\A(\C_{\Theta_h})$ to $Act$, where $\C_{\Theta_h}$ := \textsf{CreateThetaCtrl}(h)\;

				$\tau := \tau \cup \{(\s{OR}_{\phi^{i}_{in}}.\s{out} \dashrightarrow  \A(\C_{\Theta_h}).\s{set}), 
				(\A(\C).\s{out} \dashrightarrow \A(\C_{ \Theta_h}).\s{in}),  (\A(\C_{\Theta_h}).\s{out} \dashrightarrow  \s{OR}_{\varphi^{j}_{in}\vee \rho^0_{out}}.\s{in}_{\s{Index}(\chi^{z}_{in}, \varphi^{j}_{in}\vee \rho^0_{out})})\}$\; 

			}

		}
		
		\comments{

		}
		\ForEach{clause formula $\chi^0_{out}$ from DNF of $\eta_m.\rho^0_{out}$}{	
			Add an AND-gate actor $\s{AND}_{\eta_m.\chi^0_{out}}$ with $\s{size}(\chi^0_{out})$ inputs to $Act$\;
			
			\ForEach{literal $\omega_{out}$ of $\chi^0_{out}$}{	
				\textbf{let} $v_{out}$ be the variable used in $\omega_{out}$\;
				
				\If{$\omega_{out}$ equals $\neg v_{out}$ (i.e., negation is used in literal)}{

					Create $\boxed{\neg}_{Res_{v_{out}}}$  and add it to $Act$ (if not exists in $Act$)\;	
					Add ($Res_{v_{out}}.\s{output} \dashrightarrow  \boxed{\neg}_{Res_{v_{out}}}.\s{input}$) to $\tau$ (if not exists in $\tau$)\;	 
					Add ($\boxed{\neg}_{Res_{v_{out}}}.\s{output} \dashrightarrow \s{AND}_{\chi^0_{out}}.\s{in}_{\s{Index}(\omega_{out}, \chi^0_{out})}$) to $\tau$\;
				}
				\lElse
				{
					$\tau := \tau \cup  \{Res_{v_{out}}.\s{output} \dashrightarrow \s{AND}_{\chi^0_{out}}.\s{in}_{\s{Index}(\omega_{out}, \chi^0_{out})}\}$	
				}	
			}	
			$\tau := \tau \cup \{\s{AND}_{\chi^0_{out}}.\s{out} \dashrightarrow  \s{OR}_{\varphi^{j}_{in}\vee \rho^0_{out}}.\s{in}_{\s{Index}(\chi^0_{out}, \varphi^{j}_{in}\vee \rho^0_{out})}\}$
			
		}
		$\tau := \tau \cup \{\s{OR}_{\phi^{j}_{in}\vee \rho^0_{out}}.\s{out} \dashrightarrow  \A^{(m)}.\s{release}\}$\;
		
	}

}

\caption{Synthesize monitoring controllers (for pattern P1, P2, P3)}
\label{step.2.p123}

\end{footnotesize}
\end{ProcessStep}

\vspace{-2mm}
\subsection{Monitors and and Phase Adjustment Actors} 
\vspace{-1mm}

The second step of the algorithm synthesizes controllers for monitoring the appearance of an event matching the subformula, and connects these controllers to previously created actors for realizing high-level control objectives. 
For a formula $\phi$ in DNF form, let $\s{size}(\phi)$ return the number of clauses in $\phi$.  For clause formula $\chi^i_{in}$ in $\phi$, let $\s{In}(\chi^i_{in})$ return the set of all input variables and $\alpha=\s{Index}(\chi^{i}_{in}, \phi^{i}_{in})$ specify that $\chi^{i}_{in}$ is the $\alpha$-th clause in $\phi^{i}_{in}$. 



\paragraph{Step 2.1 - Realizing ``\s{input}'' part for pattern P1, P2, P3.}

In Step~\ref{step.2.p123}, from line 3 to~9, the algorithm synthesizes controller realizing the portion \s{input} listed in Table~\ref{table.specification.skeleton}, or equivalently, the $\phi^{i}_{in}$ part listed in Table~\ref{table.specification.pattern}. Line~4 first creates an OR gate, as the formula is represented in DNF. Then synthesize a controller for monitoring each clause formula (line 5, 6) using function \textsf{Syn}, with input variables defined in $\s{In}(\chi^i_{in})$ and a newly introduced output variable \{\s{out}\}\footnote{For pattern type P2 or P3, one needs to have each clause formula of $\phi^{i}_{in}$ be of form $\chi^i_{in}$, i.e., the highest number of consecutive~\x should equal~$i$. The purpose is to align $\chi^i_{in}$ with the preceding $\x^{i}$ in $\G (\phi^i_{in} \rightarrow \x^i  (\varrho_{out}\,\CU\, (\varphi^{j}_{in} \vee \rho^0_{out})) )$ or $\G (\phi^i_{in} \rightarrow \x^i  \varrho_{out})$. If a clause formula in DNF contains no literal starting with $\x^i$, one can always pad a conjunction $\x^i\,\s{true}$ to the clause formula. The padding is not needed for~P1.}. The first attempt is to synthesize $\G (\chi^i_{in} \leftrightarrow \x^{i} \sig{out})$. By doing so, the value of $\chi^i_{in}$ is reflected in \sig{out}. However, as the output of the synthesized controller is connected to the input of an OR-gate (line~8) and subsequently, passed through the port ``\s{input}'' of the high-level controller (line~9),  one needs to also ensure that from time 0 to $i-1$, $\s{out}$ remains~\s{false}, such that the high-level controller $\A_{m}$ for specification $\eta_m$ will not be ``unintentionally'' triggered and subsequently restrict the output. To this end, the specification to be synthesized is $\G (\chi^i_{in} \leftrightarrow \x^{i} \sig{out}) \wedge \bigwedge_{z=0 \ldots i-1} \x^z \neg \sig{out}$, being stated in line~6. 

For above mentioned property that needs to be synthesized in line~6, one does not need to use full LTL synthesis algorithms. Instead, we present a simpler algorithm (Algorithm~\ref{algo.synthesis.monitor}) which creates a controller in time linear to the number of variables times the maximum number of~\x operators in the formula. Here again for simplicity, each state variable is three-valued ($\s{true}$, $\s{false}$,~$\s{u}$); in implementation every 3-valued state variable is translated into~2 Boolean variables. 
In the algorithm, state variable $v_{in}[i]$ is used to store the i-step history of for $v_{in}$, and $v_{in}[i]=\s{u}$ means that the history is not yet recorded. Therefore, for the initial state, all variables are set to $\s{u}$ (line~4). The update of state variable $v_{in}[i+1]$ is based on the current state of $v_{in}[i]$, but for  state variable $v_{in}[1]$, it is updated based on current input $v_{in}$ (line 17). With state variable recording previously seen values, monitoring the event is possible, where the value of $\s{out}$ is based on the condition stated from line~6 to~16.

Consider a controller realizing  $\chi^i_{in} :=  \neg \s{in1} \wedge \x \s{in1} \wedge \x \s{in2} \wedge \x\x \neg \s{in2}$, being executed under a run prefix (\s{false, false})(\s{true, true})(\s{true, false}). As shown in Figure~\ref{fig:Monitoring},  the update of state variables is demonstrated  by a left shift. The first and the second output are \s{false}.
After receiving the third input, the controller is able to detect a rising edge of \s{in1} (via \sig{in1[2]}=\s{false} and \sig{in1}[1]=\s{true}) is immediately followed by a falling edge of \s{in2} (via \sig{in2[1]}=\s{true} and \sig{in2}=\s{false}). 



\setcounter{algocf}{0}
\begin{algorithm}[t]
	\begin{footnotesize}

\SetKwInOut{Input}{Input}
\SetKwInOut{Output}{Output}
\Input{LTL specification $\G (\chi^i_{in} \leftrightarrow \x^{i} \sig{out}) \wedge \bigwedge^{i-1}_{z=0}\x^z \neg \sig{out}$), input variables $\s{In}(\chi^i_{in})$, output variables \{\s{out}\} } 
\Output{Mealy machine $\C = (Q, q_0,  \mathbf{2}^{V_{in}}, \mathbf{2}^{V_{out}}, \Delta)$ for  realizing the specification}

$V_{out} := \{\s{out}\}$, $V_{in} := \s{In}(\chi^i_{in})$\;
\ForEach(\tcp*[h]{Create all state variables in the Mealy machine}){Variable $v_{in} \in \s{In}(\chi^i_{in})$ }
{
	\lFor{$j = 1 \ldots i$}
	{
		$Q := Q \cup \{v_{in}[j]\}$, where $v_{in}[j]$ is three-valued (\s{true}, \s{false}, \s{$\s{u}$}) 
	}
}
$q_0 := \bigwedge_{v_{in} \in \s{In}(\chi^i_{in}), j \in \{1, \ldots i\}} v_{in}[j] := \s{$\s{u}$}$  \,\, \,\, \,\, \tcp*[h]{Initial state}\;


\textbf{let} \s{Cond} := \s{true}\;
\ForEach{literal $\x^k v_{in}$ in $\chi^i_{in}$}{
	\eIf{$k=i$}
	{	\s{Cond} := $\s{Cond} \wedge (v_{in} = \s{true}) $
	}{
	\s{Cond} := $\s{Cond} \wedge (v_{in}[i-k] = \s{true}) $
}
}
\ForEach{literal $\x^k \neg v_{in}$ in $\chi^i_{in}$}{
	\eIf{$k=i$}
	{	\s{Cond} := $\s{Cond} \wedge (v_{in} = \s{false}) $
	}{
	\s{Cond} := $\s{Cond} \wedge (v_{in}[i-k] = \s{false}) $
}	
}
$\delta_{out} := (\s{out} := \s{Cond})$ \,\, \,\, \,\, \tcp*[h]{Output assignment should follow the value of \s{Cond}}\;

$\delta_{s} := (\bigwedge_{v_{in} \in \s{In}(\chi^i_{in}),  j = 1 \ldots i-1} v_{in}[j+1] := v_{in}[j]) \wedge  (\bigwedge_{v_{in} \in \s{In}(\chi^i_{in})} v_{in}[1] := v_{in})$\;
\caption{Realizing \textsf{Syn} without full LTL synthesis}
\label{algo.synthesis.monitor}

\end{footnotesize}
\end{algorithm}

\vspace{-2mm}
\paragraph{Step 2.2 - Realizing ``\s{release}'' part for pattern P2.}

Back to Step~\ref{step.2.p123}, the algorithm from line~10 to~29 synthesizes a controller realizing the portion \s{release} listed in Table~\ref{table.specification.skeleton}, or equivalently, the $\varphi^{j}_{in} \vee \rho^0_{out}$ part listed in Table~\ref{table.specification.pattern}. The  DNF structure is represented as an OR-actor (line~11), taking input from $\varphi^{j}_{in}$ (line 12-18) and  $\rho^0_{out}$ (line~19-28).

For $\rho^0_{out}$ (line~19-28), first create an AND-gate for each clause in DNF. Whenever output variable $v_{out}$ is used, the wiring is established by a connection to the \s{output} port of $Res_{v_{out}}$ (line~27). Negation in the literal is done by adding a wire to connect $Res_{v_{out}}$ to a dedicated negation actor $\boxed{\neg}_{Res_{v_{out}}}$ to negate the output (line~23 to~26). 
Consider, for example, specification $S2$ of the automatic door running example, where the ``\s{release}'' part $(\s{in1}\,\vee\,\s{in0}\,\vee\,\s{out0})$ is a disjunction of literals using output variable $\s{out0}$\@. As a consequence,
one creates an AND-gate (line~20) which takes one input $Res_{\s{out0}}.\s{output}$ (line~27), and connects this AND-gate to the OR-gate (line~28)\@. 
Figure~\ref{fig:Resolution}(a) displays an optimized version of this construction, since the single-input AND-gate may be removed and $Res_{\s{out0}}.\s{output}$ is directly wired with the OR-gate. 

For $\varphi^{j}_{in}$ (line~12 to~23), similar to Step~2.1, one needs to synthesize a controller which tracks the appearance of~$\chi^h_{in}$ (line~13). However, the start of tracking is triggered by $\phi^{i}_{in}$ (the \sig{input} subformula). That is, whenever $\phi^{i}_{in}$ is $\s{true}$, start monitoring if $\varphi^{j}_{in}$ has appeared \s{true}. This is problematic when $\chi^h_{in}$ contains $\x$ operators (i.e., $h >0$).
To realize this mechanism, at line~19, the function \textsf{CreateThetaCtrl} additionally initiates a controller which guarantees the following: Whenever input variable $\s{set}$ turns \s{true}, the following~$h$ output value of $\s{out}$ is set to \s{false}. After that, the value of output variable \s{out} is the same as the input variable \s{in}. This property can be formulated as $\Theta_h$ (to trigger consecutive $h$ $\s{false}$ value over $\s{out}$ after seeing $\s{set}=\s{true}$) listed in Equation~\ref{eq:theta}, with implementation shown in Figure~\ref{fig:Theta}. By observing the Mealy machine and the high-level transition function, one infers that the time for constructing such a controller in symbolic form is again linear to $h$.



\vspace{-2mm}
\begin{equation}\label{eq:theta}
\Theta_h := (\neg \s{out}\,\CU\,\s{set})  \wedge \G (\s{set} \rightarrow (\bigwedge^{h-1}_{z=0} \neg\x^z \s{out} \wedge \x^h ((\s{in} \leftrightarrow \s{out})\,\CU \,\s{set})))
\end{equation}
%
\noindent
The overall construction in Step~\ref{step.2.p123} is illustrated using the example in Figure~\ref{fig:Time.Shifting}, which realizes the formula
\begin{equation}\label{eq:simple}
\G ((\neg \s{in1} \wedge \x \s{in1}) \rightarrow \x (\s{out1}\; \CU (\neg \s{in2} \wedge \x \s{in2})   ))
\end{equation}
with $V_{in} = \{\s{in1}, \s{in2}\}$ and  $V_{out} = \{\s{out1}\}$\@.
This specification requires to set output \s{out1} to $\s{true}$ when a rising edge of \s{in1} appears, and after that, \s{out0} should remain \s{true} until detecting a raising edge of $\s{in2}$. Using the algorithm listed in Step~\ref{step.2.p123}, line~6 synthesizes the monitor for the \s{input} part (i.e., detecting rising edge of $\s{in1}$), line~13 synthesizes the monitor for the \s{release} part (i.e., detecting rising edge of $\s{in2}$), line~14 creates the  wiring from input port to the monitor. As $h=1$ (line~16), line~17 creates $\A(\C_{\Theta_1})$, and line~18 establishes the wiring to and from $\A(\C_{\Theta_1})$. 

The reader may notice that it is incorrect to simply 
connect the monitor controller for $\neg\s{in2} \wedge \x \s{in2}$ directly to TrUB.$\s{release}$, as, 
when both 
$\neg\s{in1} \wedge \x \s{in1}$ and $\neg\s{in2} \wedge \x \s{in2}$ are \s{true} at the same time, TrUB.\sig{output} is unconstrained. 
 On the contrary, in Figure~\ref{fig:Time.Shifting}, when $\neg\s{in1} \wedge \x \s{in1}$ is~\s{true} and the value is passed through TrUB.\s{input}, $\A(\C_{\Theta_1})$ enforces to invalidate the incoming value of TrUB.\s{release} for~1 cycle by setting it to \s{false}\@.


\begin{figure}[t]
	\centering
	\begin{minipage}{0.54\textwidth}
		\includegraphics[width=\columnwidth]{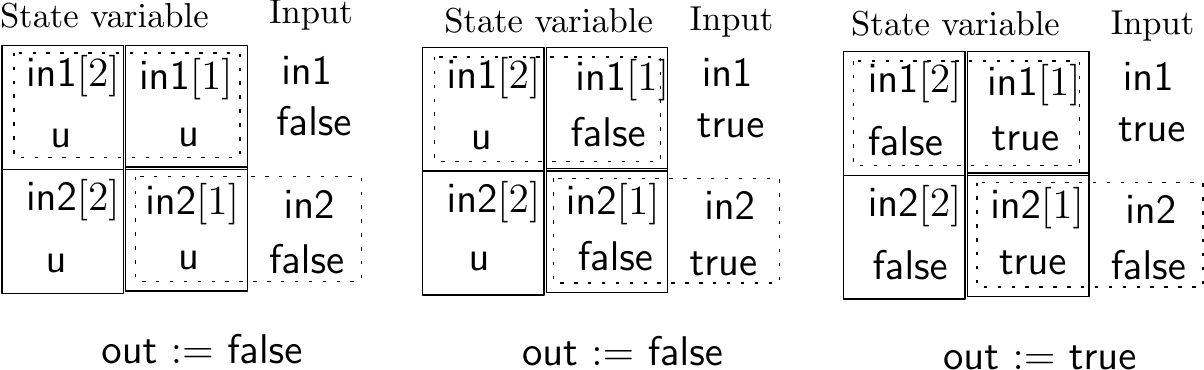}	
		\caption{Executing monitor with $\chi^i_{in} :=  \neg \s{in1} \wedge \x \s{in1} \wedge \x \s{in2} \wedge \x\x \neg \s{in2}$, by taking first three inputs (\s{false, false})(\s{true, true})(\s{true, false}).} 	
		\label{fig:Monitoring}
		\includegraphics[width=\columnwidth]{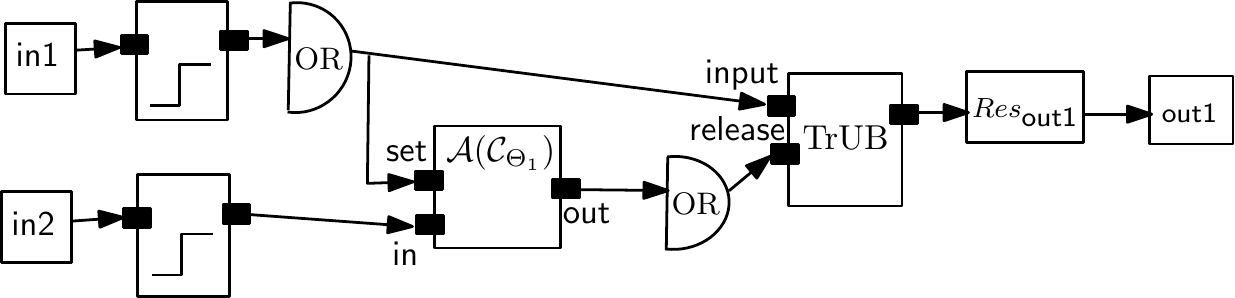}	
		\caption{Correct controller construction for specification satisfying pattern P2.}	
		\label{fig:Time.Shifting}
	\end{minipage}\hfill
	\begin{minipage}{0.45\textwidth}
		\includegraphics[width=\columnwidth]{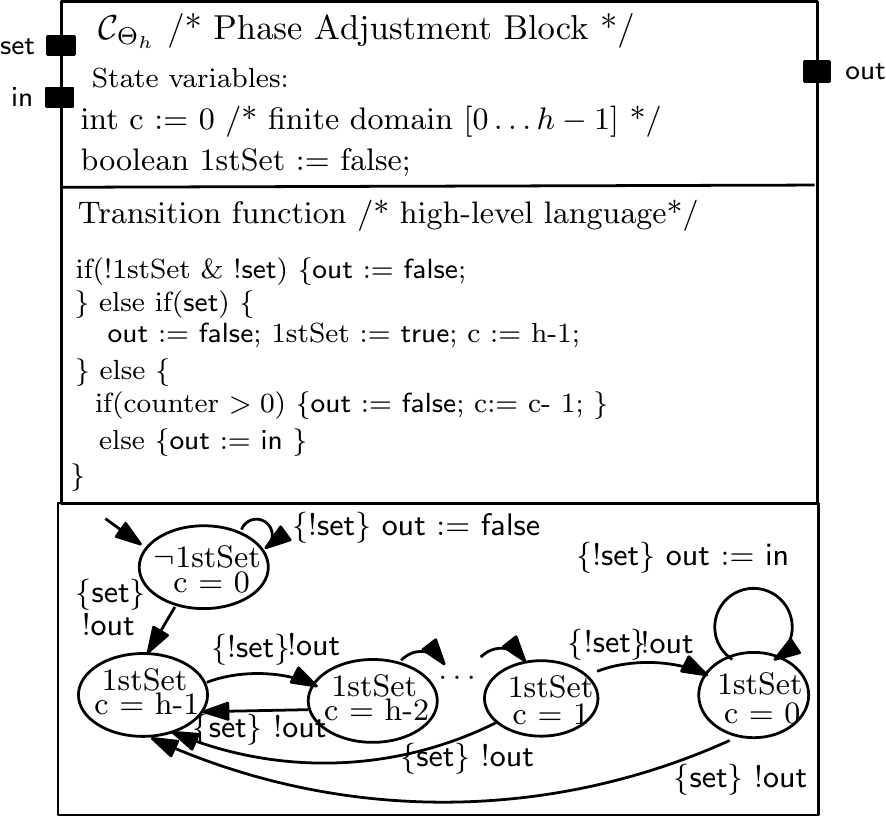}	
		\caption{Implementing $\Theta_h$ (state variables not mentioned in update remain the same value).}	
		\label{fig:Theta}
	\end{minipage}
\end{figure}



\vspace{-2mm}
\paragraph{Step 3 - Realizing "\s{input}" for pattern P4.}
For pattern P4, in contrast to pattern P1, P2, and P3, the synthesized controller (following Step~\ref{step.2.p4}) is directly connected to a Resolution Actor. To maintain maximum freedom over output variable, one synthesizes the event monitor from the specification allowing the first $i$ output to be $\s{\textendash}$, via $\bigwedge^{i-1}_{z=0}\x^{z} \s{dc} \wedge \x^{i}\G \neg \s{dc}$. The construction is analogous to Algorithm~\ref{algo.synthesis.monitor}\@.

\vspace{-2mm}
\paragraph{Optimizations.} 
Runtimes for Steps~\ref{step.2.p123} and~\ref{step.2.p4} may be optimized 
by using simple pattern matching and hashing of previously synthesized controllers. 
%
%
	We are listing three different opportunities for optimized 
	{\em generation of monitors}\@.
    First, the controller in Figure~\ref{fig:Resolution}(a) for monitoring $\neg \s{in0} \wedge \x\,\s{in0}$ is connected to two high-level controllers.
    The second case can be observed in Figure~\ref{fig:Time.Shifting}, where by rewriting $\s{in1}$ and $\s{in2}$ to $\s{in}$, the controller being synthesized is actually the same. Therefore, one can also record the pattern for individual monitor and perform synthesis once per pattern. 
	A third opportunity for optimization occurs when Algorithm~\ref{algo.synthesis.monitor} takes $i=0$ (i.e., no \x operator is used). In these case there is no need to create a controller at all and one may proceed by directly building a combinatorial circuit, similar to the constructions of line~19 to~28 in Step~\ref{step.2.p123}. For example, for specification S2 of the automatic door, the \s{release} part is $\s{in1} \vee \s{in0} \vee \s{out0}$; since no $\x$ operator occurs, a combinatorial circuit is created by wiring directly $\boxed{\s{in1}}$ and $\boxed{\s{in0}}$ to the OR-gate.

	
	
	


\vspace{-2mm}
\subsection{Parameter Synthesis for 2QBF without Unroll}\label{sub.sec.algorithm.parameter.syn.no.unroll}

Step~\ref{step.2.p4} as described above constructs actors as building blocks
and wires the actors according to the structure of the given \GXW specification.
The resulting (partial) controller, however, does not yet realize this specification as it may still contain unknowns in the resolution actors. 
Further checks are necessary, and a controller is rejected if one of the following conditions holds.
\begin{description}
	\item[(Condition 1)] The wiring forms a directed loop in the constructed actor-based controller. 
	\item[(Condition 2)] It is possible for a resolution actor $Res_{v_{out}}$ to receive \s{true} and \s{false} simultaneously. 
	\item[(Condition 3)] Outputs violate invariance conditions of pattern~P5.
\end{description}
Condition~1 is checked by means of a simple graph analysis: 
(1) let all ports be nodes and wirings be edges; (2) 
for each actor, create directed edges from each of its input port to each of its output port; (3) check if there exists a strongly connected component in the resulting graph using, for example, Tarjan's algorithm~\cite{tarjan}\@. 
%
%
 
 Conditions~2 and~3 are checked by means of creating corresponding \s{2QBF} satisfiability problems.
 Recall that each resolution actor $Res_{v_{out}}$ is parameterized with respect to the output $\s{A}$ when all incoming inputs for $Res_{v_{out}}$ are ``$\textsf{\textendash}$''\@. 
 The corresponding parameter assignment problem is encoded as a 
 \s{2QBF}\footnote{Quantified Boolean Formula with one top-level     
 	quantifier alternation.} 
 formula, where existential variables are the parameters to be synthesized, universal variables are input variables, and the quantifier-free body  is a logical implication specifying that the encoding of the system guarantees condition~2 and~3\@. 
 
 Step~\ref{step.3} shows a simplified algorithm for generating 2QBF constraints which does not perform unrolling. 
Stated in line~15, the quantifier free formula is of form $\Upsilon_{a} \rightarrow \Upsilon_{g}$, where 
$\Upsilon_{a}$ are input assumptions and system dynamics, and  $\Upsilon_{g}$ are properties to be guaranteed.


	First, unknown parameters are added to the set of existential variables~$V_{\exists}$ (line~2). All other variables are universal variables.
	Then based on the evaluation ordering of~$\Sys$, perform one of the following tasks:
	\begin{list1}
		\item When an element $\xi$ in the execution ordering $\Xi$ is a wire (line~5), we add \s{source} and \s{dest} as universal variables (as $V_{\forall}$ is a set, repeated variables will be neglected), and establish the logical constraint  $(\s{source} \leftrightarrow \s{dest})$ (lines~6 to~8)\@. 
		
		\item When an element $\xi$ in the execution ordering $\Xi$ is an actor, we use function \textsf{EncodeTransition} to encode the transition (pre-post) relation as constraints (line~11), and add all state variables (for pre and post) in the actor (recall our definition of Mealy machine is based on state variables) to $V_{\forall}$ using function  \textsf{GetStateVariable} (line~10)\@.
	\end{list1}
	%
	 $\Upsilon_{a}$ is initially set to $\varrho$ (line~1) to reflect the allowed input patterns regulated by the specification (specification type~P6)\@.
	Line~12 creates the constraint stating that no two inputs of a resolution actor should create contradicting conditions. As the number of input ports for any resolution actor is finite, the existential quantifier is only an abbreviation which is actually rewritten to a quantifier-free formula describing relations between input ports of a resolution actor. 

\setcounter{algocf}{2}
\begin{ProcessStep}[t]
	\begin{footnotesize}
		\SetKwInOut{Input}{Input}
		\SetKwInOut{Output}{Output}
		\Input{LTL specification $\phi = \varrho \rightarrow \bigwedge_{m=1\ldots k} \eta_m$,  $V_{in}$, $V_{out}$, actor-based (partial) controller $\Sys= (V_{in}, V_{out}, Act, \tau$)  and $\Map_{out}$ from Step~\ref{step.2.p123}. } 
		\Output{Actor-based (partial) controller $\Sys= (V_{in}, V_{out}, Act, \tau$) by adding more elements}
		
		\ForEach{$\eta_m$, $m=1\ldots k$}{		
			
			\If{$\emph{\s{DetectPattern}}(\eta_m)\in \{P4\}$}{
				Add a  $\s{size}(\eta_m.\phi^{i}_{in})$-input OR-gate actor $\s{OR}(\phi^{i}_{in})$  to $Act$\;
				\ForEach{ clause formula $\chi^i_{in}$ from DNF of $\eta_m.\phi^{i}_{in}$ }
				{
					
					$\C_{\chi}$ := \textsf{Syn}($(\G (\chi^i_{in}) \leftrightarrow \x^{i} \sig{out})) \wedge \bigwedge^{i-1}_{z=0} \x^{z} \s{dc} \wedge \x^{i}\G \neg \s{dc}$, $\s{In}(\chi^i_{in})$, \{\s{out}, \s{dc}\})\;
					
					Add $\A(\C_{\chi^i_{in}})$ to $Act$\;	
					
					\lForEach{$v_{in} \in \s{In}(\chi^i_{in})$}{
						$\tau := \tau \cup \{(\boxed{v_{in}} \dashrightarrow  \A(\C_{\chi^i_{in}}.v_{in}))\}$	
					}
				}
				$\tau := \tau \cup \{(\s{OR}_{\phi^{i}_{in}}.\s{out} \dashrightarrow  Res_{v_{out}}.{\s{input}_i})\}$\;
				
				\textbf{let} $v_{out}$ be the variable used in $\tau_{m}.\varrho_{out}$, $\s{ind}$ := $\Map_{out}$.get($v_{out}$).indexOf($m$)\;	
				
				\If	(\tcp*[h]{negation is used in literal}) {$\varrho_{out}$ equals $\neg v_{out}$}{
					Create a negation actor $\boxed{\neg}^{(m)}$ and add it to $Act$\;	
					Add ($\s{OR}_{\phi^{i}_{in}}.\s{out}  \dashrightarrow  \boxed{\neg}^{(m)}.\s{in}$), (		$\boxed{\neg}^{(m)}.\s{output} \dashrightarrow Res_{v_{out}}.{\s{input}_{\s{ind}}}$) to $\tau$\;
				}
				\lElse{
					Add ($\s{OR}_{\phi^{i}_{in}}.\s{out}  \dashrightarrow  Res_{v_{out}}.{\s{input}_{\s{ind}}}$) to $\tau$
				}
				
			} 
		}
		
		\caption{Synthesize monitoring controllers (for pattern P4)}
	\label{step.2.p4}	
	\end{footnotesize}
\end{ProcessStep}

\begin{ProcessStep}[t]
	\begin{footnotesize}

	\SetKwInOut{Input}{Input}
	\SetKwInOut{Output}{Output}
	\Input{LTL specification $\phi = \varrho \rightarrow \bigwedge_{m=1\ldots k} \eta_m$, input variables $V_{in}$, output variables $V_{out}$, partial controller implementation $\Sys= (V_{in}, V_{out}, Act, \tau$) with unknown parameters } 
	\Output{
		Controller implementation	$\Sys$  or ``\textsf{unknown}'' 
	}		


	\textbf{let} $\Upsilon_{a} := \varrho, \Upsilon_{g} := \s{true}$, $V_{\exists}, V_{\forall}$ := \textsf{NewEmptySet}();	
	
	\lForEach{$v_{out} \in V_{out}$}{
		$V_{\exists} := V_{\exists} \cup \{ Res_{v_{out}}.\s{A}  \}$ 
	} 
	

	\textbf{let} $\Xi$ be the evaluation ordering of $\Sys$ \;
	
	\ForEach{$\xi \in \Xi$}{
		
		\If(\tcp*[h]{$\xi$ is w wire; encode biimplication among two ports }){$\xi \in \tau$}{
			Let $\xi$ be $(\s{source} \dashrightarrow \s{dest})$\;
			$V_{\forall}.add(\s{source})$, 	$V_{\forall}.add(\s{dest})$\;	
			$\Upsilon_{a} := \Upsilon_{a} \wedge (\s{source} \leftrightarrow \s{dest})$\;	
		}
		 \Else{
		$V_{\forall}.add(\textsf{GetStateVariable}(\xi))$\;	
	
	$\Upsilon_{a} := \Upsilon_{a} \wedge (\textsf{EncodeTransition}(\xi) )$	/* $\xi \in Act$ */	 \;  
		
	}
	
}


\lFor{$v_{out} \in V_{out}$}{	
	$\Upsilon_{g} := \Upsilon_{g} \wedge (\not\exists i, j: (Res_{v_{out}}.\s{input}_{i}=\s{true}) \wedge (Res_{v_{out}}.\s{input}_{j} = \s{false}))$
}

\ForEach{$\eta_m$, $m=1\ldots k$}{		
	\lIf{$\emph{\textsf{DetectPattern}}(\eta_m) \in \{P5\}$}{
		$\Upsilon_{g} := \Upsilon_{g} \wedge \eta_m$
	}
}

\If{\emph{\textsf{Solve2QBF}}($V_{\exists}, V_{\forall},  \Upsilon_{a} \rightarrow \Upsilon_{g}$).\emph{\textsf{isSatisable}}}{		
	\Return $\Sys$ by replacing each $Res_{out}.\s{A}$ by the value of witness in 2QBF; 
}	\lElse{ \Return \s{unknown}}

\caption{Parameter synthesis by generating 2QBF constraints}
\label{step.3}

\end{footnotesize}

\end{ProcessStep}

The encoding presented in Step~\ref{step.3} does not involve unroll (it encodes the transition relation, but not the initial condition). Therefore, by setting all variables to be universally quantified, one approximates the behavior of the system dynamics without considering the relation between two successor states. Therefore, using Step~\ref{step.3} only guarantees \emph{soundness}: If the formula is satisfiable, then the specification is realizable (line~15,~16). Otherwise, \s{unknown} is returned (line~17).\footnote{Even without unroll, one can  infer relations over universal variables via statically analyzing the specification. As an example, consider two sub-specifications $S1: \G (\s{in1} \rightarrow (\s{out}\,\CU\,\s{in2}))$ and  $S2: \G (\s{in2} \rightarrow (\neg\s{out} \,\CU\, \s{in1}))$.  One can infer that it is impossible for $TrUB^{(1)}$ and $TrUB^{(2)}$ to be simultaneously have state variable $lock=\s{true}$, as both starts with $lock=\s{false}$, and if $S1$ first enters lock $(lock=\s{true})$ due to $\s{in1}$, the $S2$ cannot enter, as \s{release} part of $S2$ is also $\s{in1}$. Similar argument follows vice versa.}

As each individual specification of one of the types \{P1, P2, P3, P4\} is trivially realizable, the reason for rejecting a specification is (1) simultaneous \s{true} and \s{false} demanded by different sub-specifications,  (2) violation of properties over output variables (type P5), and (3) feedback loop within $\Sys$. Therefore, as  Steps~\ref{step.1}-\ref{step.3} guarantees non-existence of above three situations, the presented method is \emph{sound}.

\begin{theorem}{(\bf Soundness)}
	Let $\phi$ be a \GXW specification, and ${\cal S}$ be an actor-based controller as generated by Steps 1-4 from $\phi$; 
	then ${\cal S}$ realizes $\phi$\@.
\end{theorem}

\vspace{-1mm}

\begin{wrapfigure}{r}{0.3\textwidth} 
	\vspace{-10pt}
	\begin{center}
		\includegraphics[width=0.3\textwidth]{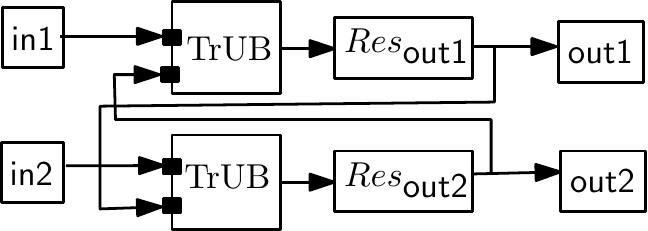}
		\caption{Incompleteness example.}
		\label{fig:feedback.loop}
	\end{center}
	\vspace{-20pt}
\end{wrapfigure} 
The \GXW synthesis algorithm as described above, however is {\em incomplete}, as controllers with feedback loops are rejected; that is, whenever output variables listed in the \s{release} part of~P2 necessitate simultaneous reasoning over two or more output variables. 
Figure~\ref{fig:feedback.loop} display an controller (with feedback loop) for realizing the specification $\G(\s{in1} \rightarrow (\s{out1} \CU \s{out2})) \wedge \G(\s{in2} \rightarrow (\s{out2} \CU \s{out1}))$\@. However,
our workflow rejects such a controller even though the given specification is
realizable. With further structural restriction over~\GXW (which guarantees no feedback loop in during construction) and by using unrolling of the generated actor-based controllers, the workflow as presented here can 
be made to be \emph{completeness}, as demonstrated
in~Section~\ref{sub.sec.properties}\@. 

\textbf{(Remarks)} Notice that in the presented algorithm, the synthesized controller does not contain a detector for checking if environment assumption (type P6) is violated. 
Practically, input assignments violating~P6 should never appear. 
If such a violation is possible, then immediately after the assumption violation and from that time onwards, the controller is allowed to produce arbitrary output assignment\footnote{In practice, if a violation is possible, it should be captured and handled properly; it is never the case that output variables are allowed to be assigned with arbitrary value.}. Here we omit details, but would like to stress that one can easily mediate such scenarios by several ways. E.g., one can append a detection actor (by building a combinational circuit out of the specification type~P6) to detect the event whether the environment assumption is violated. At the same time,  provide an additional input port on every resolution actor to override the output (i.e., a resolution actor should consider first if environment assumption is violated: if so, then continuously output~\sig{false}), and link the detection actor with the newly created port.  

Also, there are unlikely scenarios such as declaring output variables without having them used in any specification. One can mediate it by always connecting a wire from one input port to the output port of a unused output variable, without building a resolution actor. The presented algorithm here also omits details of such corner cases.

\vspace{-1mm}
\subsection{General Properties for \GXW Synthesis} \label{sub.sec.properties} 




Since unrealizability of a \GXW specification is due to the conditions (1) simultaneous \s{true} and \s{false} demanded by different sub-specifications, and (2) violation of properties over output variables (type P5)\footnote{Rejecting feedback loops on the controller structure is only a restriction of our presented method and is not the reason for unrealizability; similar to Figure~\ref{fig:feedback.loop}, feedback loop can possibly be resolved by merging all actors involving feedback to a single actor.}, one can build a counter-strategy\footnote{A counter-strategy in LTL synthesis a state machine where the environment can enforce to violate the given property, regardless of all possible moves by the controller~\cite{popl89}.}  
by first building a tree that provides input assignments to lead all runs to undesired states violating (1) or (2), then all leafs of the tree violating (1) or (2) are connected to a self-looped final state, in order to accept $\omega$-words. 
As the \s{input} part listed in Table~\ref{table.specification.skeleton} does not involve any output variable, a counter-strategy, if exists, can lead to violation of (1) or (2) within $\Omega$ cycles, where $\Omega$ is a number sufficient to let each \s{input} part of the sub-specification be \s{true} in a run.

\begin{wrapfigure}{r}{0.35\textwidth} 
	\vspace{-20pt}
	\begin{center}
		\includegraphics[width=0.35\textwidth]{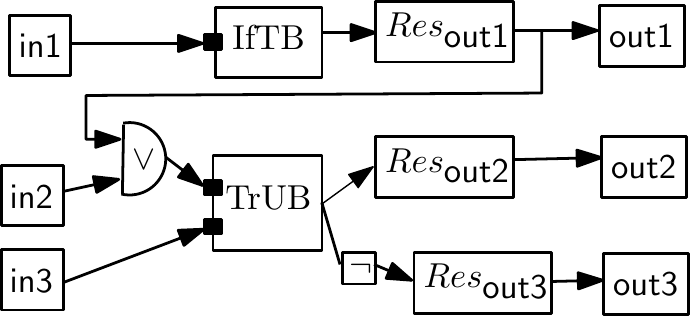}
		\caption{Control implementation.}
		\label{fig:Extension}
	\end{center}
	\vspace{-20pt}
	\vspace{1pt}
\end{wrapfigure}




\begin{lemma}\label{lemma.bound.no.output}
	\normalfont
	For \GXW specification $\varrho \rightarrow \bigwedge_{m=1\ldots k} \eta_m$, if (a) $\rho^0_{out}$ is \s{false} for all $\eta_m$ of type~P2 and (b) no specification of type P5 exists, then if the specification is not realizable, then there exists a counter-strategy which leads to violation of (1) or (2) in $\Omega$ steps, where $\Omega$ is bounded by the sum of  (i) the number of specifications $k$, and (ii)  the sum of all $i$ value defined within each $\phi^i_{in}$ of $\eta_m$.
\end{lemma}

\noindent  When $\rho^0_{out}$ is \s{false} for all $\eta_m$ of type~P2, our presented construction guarantees no feedback loop. As no specification of type P5 exists, the selection of $\s{A}$ never influences whether the specification is realizable. Therefore, quantifier alternation is removed. To this end, checking unrealizability is equivalent to nondeterministically guessing $\Omega$ input assignments and subsequently, checking if a violation of (1) or (2) appears by executing $\Sys$. This also means that under the restriction from Lemma~\ref{lemma.bound.no.output}, a slight modification of Step~\ref{step.3} to perform unrolling the computation $\Omega$-times makes our synthesis algorithm {\em complete}\@.

\begin{lemma}\label{lemma.coNP}
	\normalfont
    Deciding whether a given \GXW specification, which also obeys the additional restrictions as stated in Lemma~\ref{lemma.bound.no.output}, is realizable or not is in \s{co-NP}\@.  	
\end{lemma}


\noindent 
For the general case, the bound in Lemma~\ref{lemma.bound.no.output} remains valid (as \s{input} part is not decided by the output variable). 
Complexity result is achieved by, without using our construction, directly using finite memory to store and examine all possible control strategies in $\Omega$ steps. 


\begin{lemma}\label{lemma.bound.output}
	\normalfont
	For \GXW specification $\varrho \rightarrow \bigwedge_{m=1\ldots k} \eta_m$,   if the specification is not realizable,  then there exists a counter-strategy which leads to violation of (1) or (2) in $\Omega$ steps, where $\Omega$ is bounded by condition similar to Lemma~\ref{lemma.bound.no.output}. 
	
\end{lemma}

\vspace{-2mm}

\begin{lemma}\label{lemma.PSPACE.complete}
	\normalfont
	Deciding whether a given \GXW specification is realizable or not is in \s{PSPACE}\@.  
\end{lemma}

\noindent The above mentioned bounds are only conditions to detect realizability of a \GXW specification, while our presented workflow in Section~\ref{sec.algorithms} targets generating structured implementations. Still, by unrolling the computation $\Omega$-times, one can detect if a controller, following our regulated structure, exists.

\subsection{Extensions}
\vspace{-2mm}
One can extend the presented workflow to allow richer specification than previously presented \GXW fragment. Here we outline how these extensions are realized by considering the following sample specification: $\G (\s{in1} \rightarrow \s{out1}) \wedge \G ((\s{in2} \vee \s{out1}) \rightarrow ((\s{out2} \wedge \neg \s{out3}) \CU\,\s{in3}))$. The  \s{SDF} controller implementation is shown in Figure~\ref{fig:Extension}.
%
	First, conjunctions in $\varrho_{out}$ can be handled by considering each output variable separately. E.g., for $\varrho_{out} \equiv \s{out2} \wedge \neg \s{out3}$, in Figure~\ref{fig:Extension} both are connected to the same TrUB. 
	Second, the use of output variables in ``\s{input}'' part for pattern P1, P2, P3 is also supported, provided that in effect a combinatorial circuit is created (i.e., output variables should always proceed with $\x^i$), and the generated system does not create a feedback loop. E.g., for the antecedent $(\s{in2} \vee \s{out1})$, it is created by wiring the $Res_{out1}.\s{out}$ to an OR-gate. 

\vspace{-3mm}
\section{Experimental Evaluation}\label{sec.evaluation}
\vspace{-2mm}

We implemented a tool for $\GXW$ synthesis in \s{Java}, which invokes $\s{DepQBF}$~\cite{depqbf} (Version 5.0) for \s{QBF} solving. Table~\ref{table.experimental.result} includes experimental results for a 
representative subset of our \s{PLC} benchmark examples. 
Execution times is recorded using Ubuntu VM (Virtual Box with 3GB RAM) running on an Intel i7-3520M 2.9 Ghz CPU and 8GB RAM)\@.
Most control problems are solved in less than a second\footnote{Approximately 0.25 seconds is used for initializing JVM in every run.}.
\GXW synthesis always generated a controller without feedback loops for 
all  examples\@.




Table~\ref{table.experimental.result} lists a comparison of execution times of 
\GXW synthesis and the bounded LTL synthesis tool Acacia+~\cite{acacia12} (latest version 2.3)\@.  
We used the option \texttt{--player 1} of Acacia+ for forcing the environment to
take a first move, but we  did not do manual annotation in order to support compositional synthesis in Acacia+, as it is not needed by our tool\@.
For many of the simpler case studies, the reported runtimes of Acacia+ are similar to \GXW synthesis. However, \GXW seems to
scale much better to more complex case studies with a larger number of input and output variables such as examples 5,  9, 11, 12, 13, 15, 16, 17, 18, 19 in Table~\ref{table.experimental.result}\@.
The representation of the generated controller in terms of a system of interacting actors in \GXW synthesis, however,
allows the engineering to trace each sub-specification with corresponding partial implementation. In fact the structure of the  controllers generated by \GXW is usually similar to reference implementations by the case study providers.
In contrast,  a controller  expressed in terms of single Mealy machine is
rather difficult to grasp and to maintain for problems such as example 18 with 13 input and 13 output variables.




\section{Related Work} \label{sec.related}
\vspace{-2mm}

Apart from the description in Section~\ref{sec.introduction}, here we  compare \GXW synthesis with related \s{GR(1)} synthesis (e.g.,~\cite{anzu,gr1,gr1robots,ratsy}) 
and bounded LTL synthesis (e.g,~\cite{acacia12,Ehlers11,ScheweF07a}) techniques\@.

Synthesis for the \s{GR(1)} fragment of \s{LTL} is in time polynomial to the number of nodes of a generated game, which is \s{EXPTIME} when considering exponential blow-up caused by input and output variables. \GXW is in \s{PSPACE}\@, where \GXW allows $\CU$ and $\s{GR(1)}$ allows $\F$.
Even though it has been demonstrated that the expressiveness of \s{GR(1)} is enough to cover many practical examples,
the use of an until logical operator, which is not included in \s{GR(1)}, proved to be essential for 
encoding a majority of our \s{PLC} case studies.
Also, implementations of \s{GR(1)} synthesis such as Anzu~\cite{anzu} do not generate structured controllers. 
Since \s{GR(1)} synthesis, however, includes a round-robin arbiter for circulating among sub-specifications, 
the systematic structuring of controllers underlying \GXW synthesis may be applicable for synthesizing structured \s{GR(1)} controllers\@.


Bounded synthesis supports full LTL and is based on a translation of the LTL synthesis problem to safety games. By doing so, one solves the safety game and finds smaller controllers (as demonstrated in synthesis competitions via tools like Simple BDD solver~\cite{comptition2014}, AbsSynthe~\cite{AbsSynthe}, Demiurge~\cite{demiurge}).
The result of solving safety games in bounded \s{LTL} synthesis usually is a monolithic Mealy (or Moore) machine, 
whereas our \GXW synthesis method of creating  \s{SDF} actors may be understood as a way of avoiding the expensive
construction of the product of machines.
Instead, we are  generating controllers by means of  wiring smaller sub-controllers for specific monitoring and 
event triggering tasks.
The structure of the resulting controllers seem to be very close to what is happening in practice, as a number of
our industrial benchmark examples are shipped with reference implementation which are usually structured in a similar way. The size of the representations of generated controllers is particularly important when
considering resource-bounded embedded computing device such as a \s{PLC}s.
\s{LTL} component synthesis, however, has the same worst-case complexity 
as full \s{LTL} synthesis~\cite{synthesiscomponent}\@.

\begin{table}[t]
	\centering
\begin{scriptsize}
		\begin{tabular}{|c|l|l|c|c|c|}
			\hline
			ID & Description & Source & I/O vars &   \GXW Time (s) & Acacia+ Time (s) \\ \hline
			1	& Automatic Door & Ex15 \cite{source1blogspot} & (4,3) & 0.389 & 0.180 \\ \hline
			2	& Simple Conveyor Belt & Ex7.1.19 \cite{source3petry} & (3,3) & 0.556 & 0.637 \\ \hline
			3	& Hydraulic Ramp & Ex7.1.3\cite{source3petry} & (5,2)  & 0.642 & 0.451\\ \hline
			4	& Waste Water Treatment V1 & Ex7.1.8 \cite{source3petry} & (6,3)  & 0.471 & 0.323 \\ \hline
			5	& Waste Water Treatment V2 & Ex7.1.9 \cite{source3petry} & (8,9) & 0.516 & 5.621 \\ \hline
			6	& Container Fusing & Ex10 \cite{source1blogspot}  & (7,6) & 0.444 & 0.425 \\ \hline
			7	& Elevator Control Mixing Plant & Ex7.1.4 \cite{source3petry} & (10,5) & 0.484 & 2.902 \\ \hline
			8	& Lifting Platform & Ex21 \cite{source2kaftanABB} & (6,3) & 0.350 & 0.645 \\ \hline
			9	& Control of Reversal & Ex36 \cite{source2kaftanABB} & (7,7) & 0.395 & 2.901 \\ \hline
			10	& Gear Wheel & Ex19 \cite{source2kaftanABB} & (4,6) &  0.447 & 0.302  \\ \hline
			11	& Two Directional Conveyor (simplified) & Ex7.1.31.1 \cite{source3petry} & (9,5) & 0.789 & 6.552 \\ \hline
			12	& Garage Door Control & Ex7.1.25 \cite{source3petry} & (13,5) & 0.574 & 7.002 \\ \hline
			13	& Contrast Agent Injection & Ex7.1.18 \cite{source3petry} & (6,8) & 0.458 & 3.209 \\ \hline
			14	& Identification & Ex39 \cite{source2kaftanABB} & (5,5) & 0.430 & 0.392 \\ \hline
			15	& Monitoring Chain Elevator & Ex7.1.15\cite{source3petry} & (10,9) & 0.429 & 9.647 \\ \hline
			16	& Two Directional Conveyor  & Ex7.1.31.1 \cite{source3petry} & (12,5) & 0.890 & 51.553 \\ \hline
			17	& Control of single torque drive (simplified) & Ex7.1.26 \cite{source3petry} & (12,8) & 0.538 & 38.010\\ \hline
			18	& Gravel transportation via 3 conveyors (simplified) & Ex7.1.31.4  \cite{source3petry} & (13,13) & 1.227 & $>600$ (t.o.)\\ \hline
			19	& Control of two torque drives (simplified) & Ex7.1.26  \cite{source3petry} & (22,16) & 0.790 & $>600$ (t.o.)\\ \hline
		\end{tabular}
		\vspace{1mm}
		\caption{Experimental Result}
		\label{table.experimental.result}
		\vspace{-2mm}
\end{scriptsize}
\end{table}

\vspace{-1mm}
\section{Conclusion}\label{sec.conclusion}
\vspace{-0mm}

We have identified a useful subclass~\GXW of \s{LTL} for specifying a 
large class of embedded control problems, and we developed a novel
synthesis algorithm (in \s{PSPACE}) for automatically generating 
structured controllers in a high-level programming language with
synchronous dataflow without cycles. 
Our experimental results suggest that \GXW synthesis scales 
well to industrial-sized control problems with around~20 
input and output ports and beyond.

In this way, \GXW synthesis can readily be integrated with industrial
design frameworks such as CODESYS~\cite{codesys}, Matlab Simulink, 
and Ptolemy~II, and the generated SDF controllers (without cycles) can be statically scheduled and implemented on single and multiple processors~\cite{sdfscheduling}\@.
It  would also be interesting to use our synthesis algorithms to automatically generate control code from established requirement frameworks for embedded 
control software such as EARS\cite{ears}\@.
Moreover, our presented method supports traceability between specifications and 
the generated controller code as required by safety-critical applications.
Traceability is also the basis for an incremental development methodology. 

One of the main impediments of using synthesis 
in engineering practice, however, is the lack of useful and automated feedback
in case of unrealizable specifications~\cite{assumption,DBLP:conf/cav/ChengHRS14,specification.mining} or realizable specifications with unintended realizations. 
The use of a stylized specification languages such as \GXW seems to 
be a good starting point 
for supporting design engineers in identifying and analyzing 
unrealizable specifications, since there are only a relatively
small number of potential sources of unrealizability 
in \GXW specifications\@\footnote{Observe the example in the paper, \s{closing\_stopped} contains \s{out0}. This is the only part that is not mentioned in the textural specification, but it is required to make the specification realizable. Introducing \s{out0} needs creativity, and it is the part where one needs an “engineer in the loop”. Not mentioned in this paper, we are also developing concepts in order to automatically add such a disjunction, similar to discovery of environment assumptions as investigated by us and also by others (e.g.,~\cite{specification.mining}).}. Finally, hierarchical \s{SDF} may
also be useful for modular synthesis~\cite{sdfcomposition}\@.

\subsection*{Acknowledgement}
We thank L\u acr\u amioara~A\c stef\u anoaei for her feedback during the development of the paper.








\appendix

\noindent \begin{Large}{\textbf{Appendix}}\end{Large}


\section{Operational Semantics of \s{SDF}}

The operational semantics of \s{SDF} can be summarized using the below action sequence; one can use the example in Figure~\ref{fig:Actor.Composition}(a)(b) to ease understanding.  
\vspace{3mm}

\begin{minipage}{0.95\textwidth}
		\begin{small}
			\begin{enumerate}[(i)]
				\item All ports start with initial value \s{undefined}; 
				\item A cycle is started by reading inputs andb setting the
 external input ports to either \s{true} or \s{false};
	\item For all wires connected to the same source (external input port / internal output port), copy data to the
 connected destination (internal input port / external output port). Lastly, change the value for the source port to be
 \s{undefined}.   
	\item When values of all input ports of an actor are not \s{undefined}, 
 produce output and update to the corresponding output ports, by executing the
 underlying Mealy machine of the actor.
	\item When all external output ports are \s{true} or \s{false} and all
 internal ports are \s{undefined}, proceed to Step~(vi). Otherwise, continue with  Step~(iii). 	
	\item Produce output based on the data in the external output port,
 where each port can only be \s{true} or \s{false} following Step~(v). 
 Reset each external output port to \s{undefined}, and move to Step~(ii).
 \end{enumerate}
	\end{small}
	\end{minipage}

\section{Soundness}\label{sub.sec.correctness.no.unroll} 

We prove that the if a controller $\Sys$ is produced following the workflow from Step~\ref{step.1} to Step~\ref{step.3}, then it is correct, meaning that it realizes the \GXW specification $\varrho \rightarrow \bigwedge_{m=1\ldots k} \eta_m$.

The correctness proof can be understood using the following structure: (A) Prove that the behavior the controller is well-defined, i.e., given any (infinite) input sequence, the controller can generate an infinite output sequence such that the combined sequence forms an $\omega$-word. (B) As the specification under synthesis has the structure $ \varrho \rightarrow \bigwedge_{m=1\ldots k} \eta_m$ where $\varrho$ is a property over input variables, it suffices to prove individually that all produced $\omega$-words satisfy $\eta_m$. For each $\eta_m$, we then prove that the created partial dataflow model realizes $\eta_m$. We use operational semantics to discuss the data transfer in each cycle; an alternative method is to view the data processing in each cycle analogous to applying functional composition.

\paragraph{(A)} The internal dataflow of a synthesized controller, due to the sanity check of Condition~1 in Section~\ref{sub.sec.algorithm.parameter.syn.no.unroll}, obeys the following structure (here we omit the logic gates): external input ports $\Rightarrow$  monitor controllers $\Rightarrow$ high-level controllers $\Rightarrow$ resolution actors $\Rightarrow$ external output ports. 


The satisfiability of 2QBF  guarantees that for each output variable $v_{out}$, under any input assignment,  $Res_{v_{out}}$ cannot receive from two input ports \s{true} and \s{false}. Therefore, the output value of $Res_{v_{out}}$ is well-defined, and as $Res_{v_{out}}.\s{out} \dashrightarrow \boxed{v_{out}}$, the value of $\boxed{v_{out}}$ at the end of a cycle is either updated to $\s{true}$ or to $\s{false}$. 

\paragraph{(B, Specification Type~5)} For specification $\eta_m$ which is of type~5 (invariance condition), they are guaranteed by line~15 of Step~\ref{step.3}. 

\paragraph{(B, Specification Type~3)} In Step~\ref{step.2.p123}, the algorithm synthesizes the monitor controller $\G (\chi^i_{in} \leftrightarrow \x^{i} \sig{out}) \wedge \bigwedge_{k=0 \ldots i-1} \x^k \neg \sig{out}$ and connect the controller to external input ports. The output of the monitor is connected to an OR-gate, which then connects to the input port ($\s{input}$) of the high-level controller realizing $\G (\s{input} \rightarrow \s{output})$. The output  of the high-level controller is connected to
one of the input ports of $Res_{v_{out}}$ (when $\varrho_{out} \equiv \neg v_{out}$, a negation actor is inserted in between), which then produces output to $v_{out}$. 	


Our goal is to prove that the composition of these sub-controllers via wiring of ports is a controller realizing $\G (\phi^i_{in} \rightarrow \x^i \varrho_{out})$. 
Based on the definition, the synthesized controller realizes $\G (\phi^i_{in} \rightarrow \x^i \varrho_{out})$, if for all produced $\omega$-words satisfies the property $\G (\phi^i_{in} \rightarrow \x^i \varrho_{out})$, i.e.,

\begin{description}
	\item[(\G)]  for all $j\geq 0$,  $\phi^i_{in} \rightarrow \x^i \varrho_{out}$ holds 
	\item[($\rightarrow$)] for all $j\geq 0$, if $\phi^i_{in}$ holds then $\x^i \varrho_{out}$ holds
	\item[($\x$), proof goal]  for all $j\geq 0$, if $\phi^i_{in}$ holds at time~$j$, then $\varrho_{out}$ holds at $j+i$
\end{description}

Based on the definition, $\varrho_{out}$ is either $v_{out}$ or $\neg v_{out}$, where $v_{out}$ is an output variable. Here we prove only for case where $\varrho_{out} \equiv v_{out}$; for case $\varrho_{out} \equiv \neg v_{out}$ only a negation actor is introduced and the proof is similar.
For the synthesized controller:

\begin{description}
	\item[[time $j$, monitor]] The specification of the monitor contains $\G (\chi^i_{in} \leftrightarrow \x^i \s{out})$. Thus it guarantees that for all $j\geq 0$, if $\chi^i_{in}$ holds (does not hold) at time~$j$, then at $j+i$, the  value of $\s{out}$ is computed to $\s{true}$ ($\s{false}$) after executing the monitor.

	\item[[time $j$, OR gate]] All monitors are connected to an OR gate, mimicking the formula structure. It guarantees that for all $j\geq 0$, if $\phi^i_{in}$ holds at time~$j$ (due to one of its sub-formula $\chi^i_{in}$ being $\s{true}$ at time~$j$), then output port $out$ of the OR gate is $\s{true}$ at $j+i$ (as the port \s{out} of the monitor is $\s{true}$ at time $j+i$, and it is wired to an input of the OR-gate). 
	
	Similarly, if $\phi^i_{in}$ does not hold at time~$j$ (due to all of its sub-formula $\chi^i_{in}$ being $\s{false}$ at time~$j$), then output port $out$ of the OR gate has value \s{false} at $j+i$ (as every port \s{out} of the monitor is $\s{false}$ at time $j+i$, and it is wired to an input of the OR-gate).

	\item[[time $j+i$, high-level controller] ] The high-level controller IfTB realizes $\G (\s{input} \rightarrow \s{output})$, which guarantees that at time $j+i$, if $\s{input}$ holds then $\s{output}$ holds. As the output port of the OR-gate is wired to IfTB.\s{input},  if $\phi^i_{in}$ holds at time~$j$, then at time $j+i$ IfTB.\s{output} is updated to $\s{true}$, after executing IfTB.
	
	\item [[time $j+i$, $Res_{v_{out}}$, proof goal achieved]] At time $j+i$, whenever a resolution actor $Res_{v_{out}}$ receives a \s{true} from the input, the execution of $Res_{v_{out}}$ produces the \s{true} to $Res_{v_{out}}.\s{out}$, thus for output port $v_{out}$, it is updated to $\s{true}$ at time $j+i$. As $\varrho_{out} \equiv v_{out}$, $\varrho_{out}$ holds at time $j+i$.
\end{description}

\paragraph{(B, Specification Type~4)} Proof similar to type~3, by first decomposing the specification formula, followed by showing that the dataflow guarantees desired behavior.

\paragraph{(B, Specification Type~1)} (Informal sketch; to ease understanding) Again the proof follows the structure of decomposing the specification formula,  followed by showing that the dataflow guarantees desired behavior. As the controller monitoring event $\phi_{in}$ is connected to the InUB block, to guarantee correctness, before observing event $\phi_{in}$, InUB should always set output to \s{true}. In line~6 of Step~\ref{step.2.p123}, the first~$i-1$ output of the monitor is $\s{false}$. Therefore, InUB does not change to "\textsf{\textendash}($\s{\textendash}$)" within time $0$ to $i-1$. From time~$i$ onwards,  $\phi_{in}$ is well defined and correctness is guaranteed. 

\vspace{2mm}
\noindent (Formal argument) Here we again consider $\varrho_{out} \equiv v_{out}$, as the other case ($\varrho_{out} \equiv \neg v_{out}$) is analogous. 
Based on the definition, the synthesized controller realizes $\varrho_{out} \CU \phi^i_{in}$, if for all produced $\omega$-words satisfies the property $\varrho^0_{out} \CU \phi^i_{in} \equiv (\varrho^0_{out} \U \phi^i_{in}) \vee \G\, \varrho^0_{out}$, i.e., it satisfies $(v_{out}\,\U\,\varphi^{i}_{in})$ or $\G\, v_{out}$.
We now  consider whether an input sequence makes $(\varphi^{i}_{in})$ holds or not - this condition partitions all input sequences to two categories, and subsequently, make either left of right part of the disjunction hold: 

\begin{description}
	\item [($\vee $, left)] Assume there $\exists j$ such that $\varphi^{i}_{in}$ first holds at $j$, we prove that an  implementation following the construction guarantees that  $v_{out}$ is \s{true} from 0 to $j-1$, thereby satisfying the strong-until part (\U). 
	
	As $\varphi^i_{in}$ uses consecutive $i$ \x operators, although at cycle $j$
	the output can no longer produce output $\s{true}$, as the controller cannot perform clairvoyance over future inputs, it needs to continuously output \s{true} until cycle $j+i$, such that it can decide $\chi^i_{in}$ holds at time $j$. In other words, as $j$ is the first time where $\chi^i_{in}$ is \s{true} and the controller can only know it at time $j+i$, the output should always be \s{true} from time $0$ to $j+i-1$, in order to satisfy the formula. 
	

	\begin{list1}
		\item From cycle $0$ to $i-1$, the output is always $\s{true}$. 
		Consider the monitor component. Within cycle $0$ to $i-1$, it produces \s{false}, due to ``$\bigwedge_{z=0 \ldots i-1} \x^z \neg \sig{out}$'' in realizing the specification $\G (\chi^i_{in} \leftrightarrow \x^{i} \sig{out}) \wedge \bigwedge_{z=0 \ldots i-1} \x^z \neg \sig{out}$. As all monitors are connected to an OR-gate, the output of produced by the OR-gate is \s{false} from cycle 0 to $i-1$. As the output of the OR-gate  is connected to InUB which turns ``-'' only after it receives $\s{true}$, the output of InUB is \s{true} from 0 to $i-1$. Thus, $\boxed{v_{out}}$ is updated with value $\s{true}$ from 0 to $i-1$.
		
		\item From time $i$ to $j+i-1$, due to ``$\G (\chi^i_{in} \leftrightarrow \x^{i} \sig{out})$'' in realizing the specification $\G (\chi^i_{in} \leftrightarrow \x^{i} \sig{out}) \wedge \bigwedge_{z=0 \ldots i-1} \x^z \neg \sig{out}$, as $\chi^i_{in}$ first holds at $j$, so $\s{out}$ first holds at time $j+i$, meaning that before time $j+i$, the monitor produces \s{false}. As all monitor produces \s{false} before time $j+i$, the input to the InUB is \s{false}. Thus, the output produced by InUB is $\s{true}$. Therefore, before time $j+i$, $\boxed{v_{out}}$ is always updated with value $\s{true}$, making the strong-until condition hold.

	\end{list1}

	\item [($\vee $, right)] Assume $\not\exists j$ such that $\varphi^{i}_{in}$ holds at $j$,  we prove that an  implementation following the construction guarantees that   $v_{out}$ remains $\s{true}$, thereby satisfying the Global part (\G). 
	
	\begin{list1}
		
		\item For time 0 to $i-1$, the monitor component generates \s{false}, due to ``$\bigwedge_{z=0 \ldots i-1} \x^z \neg \sig{out}$'' in realizing the specification $\G (\chi^i_{in} \leftrightarrow \x^{i} \sig{out}) \wedge \bigwedge_{z=0 \ldots i-1} \x^z \neg \sig{out}$. 
		\item For time $i$ onwards, the output of a monitor is governed by whether $\chi^i_{in}$ is \s{true} or \s{false}. As $\not\exists j$ such that $\varphi^{i}_{in}$ holds at $j$, and $\chi^i_{in}$ is a formula in the DNF of $\varphi^{i}_{in}$, $\not\exists j$ such that $\chi^{i}_{in}$ holds at $j$, thus the output of each monitor is always \s{false}. 
		
		\item For InUB, it produces ``-'' after receiving a \s{true} in its input port. Before that, it produces $\s{true}$ in its output port. As InUB never receives input with value $\s{true}$, the generated output value is always $\s{true}$. Due to the dataflow, $v_{out}$ is always $\s{true}$. Thus $\G v_{out}$ holds.
	\end{list1}
\end{description}

\paragraph{(B, Specification Type~2)} (Informal sketch; to ease understanding) The proof strategy is similar. Essentially the property can be viewed as Type~3 where the right-hand part of the implication is replaced/nested by a Type~1 specification. When the triggering of left-hand part (the \s{input} part) appears at time $t$, based on the formula one needs to ``start'' the monitor which constitutes the \s{release} part. As ``starting a monitor'' is not possible, an alternative is to perform proper reset such that it restart like new. 
That is, for monitoring \s{release} formula $\chi^h_{in}$, if at time $t$ where the triggering event $\phi_{in}$ turns \s{true}, let first~$h$ output to be \s{false}), such that the TrUB.\s{release} is not wrongly polluted by the first $h-1$ values provided by the release monitor during time $t, t+1, \ldots, t+h-1$. 
As our monitor is designed not to take reset signals (this is to facilitate monitor reuse among multiple specifications), one can alternatively achieve the same effect by the introduction of~$\C_{\Theta_h}$, which contains the mechanism to set consecutive $h$ outputs to be \s{false}, whenever its input port \s{set} receives value \s{true}.

\section{General Properties for \GXW Synthesis} \label{appendix.sec.completeness} 

As each individual specification of  \{P1, P2, P3, P4\} is trivially realizable, the reason that lead to unrealizability is (1) simultaneous \s{true} and \s{false} demanded by different sub-specifications, (2) violation of properties over output variables (type P5), which are invariance properties over output variables. Notice that our method, as it generates structured controller, can also report \s{unknown} when there exists a feedback loop in the constructed system, i.e., when output variables listed in the \s{release} part of~P2 create a need for simultaneous reasoning over two or more output variables. It is only a restriction imposed on the controller structure and is not the reason for unrealizability. 

Therefore, as unrealizability of a \GXW specification is due to (1) and (2), one can construct a counter strategy by first constructing a tree which provides input assignments that lead to undesired states violating (1) or (2), then all tree leaves violating (1) or (2) are connected to a self-looped final state, in order to accept $\omega$-words created by inputs and outputs. As the \s{input} part listed in Table~\ref{table.specification.skeleton} does not involve any output variable, a counter-strategy, if exists, can lead to violation of (1) or (2) within $\Omega$ cycles, where $\Omega$ is a number sufficient to let each \s{input} part of the sub-specification be \s{true} in a run.

\setcounter{lemma}{0}
\begin{lemma}
	\normalfont
	For \GXW specification $\varrho \rightarrow \bigwedge_{m=1\ldots k} \eta_m$, if (a) $\rho^0_{out}$ is \s{false} for all $\eta_m$ of type~P2 and (b) no specification of type P5 exists, then if the specification is not realizable, then there exists a counter-strategy which leads to violation of (1) or (2) in $\Omega$ steps, where $\Omega$ is bounded by the sum of  (i) the number of specifications $k$, and (ii)  the sum of all $i$ value defined within each $\phi^i_{in}$ of $\eta_m$.
\end{lemma}

\noindent (Note) With the constraint where $\rho^0_{out}$ is always \s{false}, it is impossible to create feedback loops during the construction, as feedback loops are created due to the connecting output to the monitoring subsystem which corresponds to the \s{release} part of a formula. 

\begin{proof}
	
	When $\rho^0_{out}$ is \s{false} for all $\eta_m$ of type~P2, then for type~P2 specification, locking the output to $\varrho_{out}$ and the release of the locking is completely determined by two input events $\phi^i_{in}$ and $\varphi^{j}_{in}$. Similarly, for specifications of type P1, P3, P4, whether output is locked to $\varrho_{out}$ is decided by $\phi^i_{in}$. 
	
	Consider the simplest case with two specifications $\tau_1 := \G (\phi^{i_1}_{in,1} \rightarrow \x^{i_1}  (\varrho_{out,1}\,\CU\, \varphi^{j_1}_{in,1}))$ and $\tau_2 := \G (\phi^{i_2}_{in,2} \rightarrow \x^{i_2}  (\varrho_{out,2}\,\CU\, \varphi^{j_2}_{in,2}))$. Let $\varrho_{out,1}$ be $\s{out1}$ and $\varrho_{out,2}$ be  $\neg\s{out1}$. Thus these two specifications can lead to conflicts ($\tau_1$ demanding $\s{out1}$ to \s{true} while $\tau_2$ demanding $\s{out1}$ to \s{false}).  Consider a finite time window of size $(i_1+1)+(i_2+1)$. It is sufficiently large to first (from time 0 to~$i_1$) make $\phi^{i_1}_{in,1}$ \s{true}, and subsequently (from time $i_1+1$ to $i_1+1+i_2$), let $\phi^{i_2}_{in,2}$ be $\s{true}$.
	Then if $\varphi^{j_1}_{in,1}$ is always \s{false} in between (from time $i_1$ to $i_1+1+i_2$), a counter strategy is created within length $(i_1+1)+(i_2+1)$. The overall concept is demonstrated in Figure~\ref{fig:TimeLine}. 

	\begin{figure}[h]
		\centering
		\includegraphics[width=0.5\columnwidth]{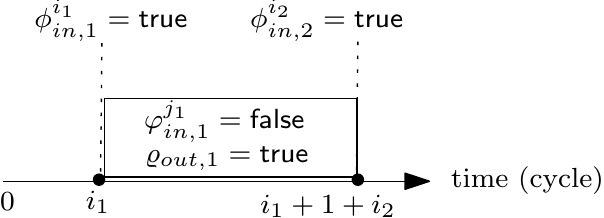}	
		\caption{A timeline describing the time windows required to produce conflict at time $i_1 + 1 + i_2$.}	
		\label{fig:TimeLine}
	\end{figure}

	For the case in Figure~\ref{fig:TimeLine}, the sufficient and necessary  condition for producing conflict is to keep  $\varphi^{j_1}_{in,1}$ always \s{false} from time $i_1$ to $i_1+1+i_2$. 
	That is, if within time $[i_1, i_1+1+i_2]$, when the environment provides a sequence of inputs to make $ \phi^{i_2}_{in,2}$ \s{true}, the sequence will also make $\varphi^{j_1}_{in,1}$ \s{true}, then it is impossible to create a counter strategy. Similarly, one can reverse the ordering of events by first making $\phi^{i_2}_{in,2}$ \s{true} followed by making $\phi^{i_1}_{in,1}$ \s{true}; the  sufficient and necessary   condition for producing conflict is to keep  $\varphi^{j_2}_{in,2}$ to \s{false} from time $i_2$ to $i_2+1+i_1$. In both cases, a time horizon $(i_1+1)+(i_2+1)$ is sufficient to demonstrate the existence of a counter-strategy. 
	
	One can also replace $\tau_1$ and $\tau_2$ by any specification pattern in $\{P1, P2, P3, P4\}$, and the bound $(i_1+1)+(i_2+1)$ is still sufficient. Lastly, by generalizing the result to $k$ specifications, we have derived the bound  $(i_1+1)+(i_2+1)+\ldots +(i_k+1)$.
	
	
\end{proof}

Lemma~\ref{lemma.bound.no.output} is based on the premise where no specification  is of type~P5, while \s{input} and \s{release} parts in Table~\ref{table.specification.skeleton} are only controlled by input variables. When no constraints are imposed to output variables due to input events, an implementation can freely select output variable assignments, as it can neither influence \s{release} nor violate properties type~P5. Therefore, when checking the existence of a counter-strategy, there is no need to use quantifier alternation, thereby making the synthesis problem easier. 


\begin{lemma}
	\normalfont
	For a \GXW specification under constraint of Lemma~\ref{lemma.bound.no.output}, deciding whether the specification is realizable is in \s{co-NP}. 	
\end{lemma}

\begin{proof}
	
	\vspace{1mm}
	
	\noindent \textbf{(A)} We first argue that the construction after Step~\ref{step.1},~\ref{step.2.p123}, and~\ref{step.2.p4}, whether an input of a resolution block is to \s{true} or \s{false} (i.e., not "\textsf{\textendash}") is only controlled by input events. 
	
	\begin{list1}
		
		\item For specification P3, the high-level controller outputs $\s{\textendash}$ once when its input receives a \s{false}, but monitor component only sends \s{true} when $\phi^i_{in}$ turns $\s{true}$. Therefore, an input of a resolution block is to \s{true} or \s{false} (i.e., not "$\s{\textendash}$") only when needed. 
		
		\item For specification P1, the high-level controller continuously outputs $\s{\textendash}$ once when its input receives a \s{true}, and before that, it outputs \s{true}. But monitor component only sends \s{true} when $\phi^i_{in}$ turns $\s{true}$, and before that (from 0 to $i-1$), the monitor sends \s{false}.

		\item Specification P2 is a combination of P3 and P1. 
		
		\item For specification P4, the monitor component is connected directly to Resolution Actor; it only sends \s{true} or \s{false} when $\phi^i_{in}$ turns $\s{true}$ or \s{false} and before that (from 0 to $i-1$), the monitor component sends $\s{\textendash}$ (\s{dc} = \s{true}). 
		
	\end{list1}	
	
	\vspace{1mm}
	
	\noindent \textbf{(B)} Using the result in Lemma~\ref{lemma.bound.no.output}, one can perform a bounded unroll over the generated actor system from Step~\ref{step.1},~\ref{step.2.p123}, and~\ref{step.2.p4}, in order to check if there exists a counter-strategy. This is because for the construction guarantees (A), therefore, whenever a counter-example is demonstrated, it is completely due to the sequence of input events. 
	The unroll algorithm is stated in Step~\ref{step.bounded.unroll}; it is similar to Step~\ref{step.3}, with the slight variation to encode the initial condition (line 21-23) and to encode the transition of $\Omega$ steps using index $0,\ldots,\Omega$ (line~5 for using~$\alpha$ to iterate over $0,\ldots,\Omega$). 
	
	As for each output variable $v_{out}$, the existential variable $Res_{v_{out}}.\s{A}$ can neither influence \s{release} nor violate properties type~P5 (due to the restriction stated in Lemma~\ref{lemma.bound.no.output}), one can simply set it to~\s{true}. Therefore, the constraint system with one quantifier alternation is simplified to checking the validity of a quantifier-free Boolean formula, or equivalently, checking the existence of an assignment to make the quantifier-free formula \s{false}.


	As running one cycle requires time linear to the number of actors and wires, which is bounded by length of the complete specification formula, deciding whether the specification is unrealizable is in \s{co-NP}.  	
	
	\vspace{1mm}	
	
	\noindent \textbf{(C)}  The result follows the fundamental property of LTL synthesis - an LTL specification is either realizable or unrealizable. As deciding whether the specification is unrealizable is in \s{co-NP}, the dual problem of  whether the specification is realizable can be decided in time \s{co-NP}.

	\vspace{1mm}
	\noindent \textbf{(Remark)} We can also analyze the number of variables and the number of clauses created in each unroll. Those numbers are timed with $\Omega=(i_1+1)+(i_2+1)+\ldots +(i_k+1)$ to derive the total number of variables and clauses. 
	
	\begin{list1}
		\item 	For all global input and output ports, they are encoded as variables. 
		
		\item For each resolution block, it takes at most $k$ inputs, so at most $2k$ variable is needed (a factor of $2$ is due to the use of $\s{\textendash}$). The output of a resolution block can be syntactically replaced by the corresponding global output port.
		
		\item 	For each specification $\eta_m$:
		\begin{itemize}
			\item Implementing each $\chi^i_{in}$ of  $\phi^i_{in}$  uses variables of size $i$ times the number of input variables in $\chi^i_{in}$. For type~P4, one adds a counter for counting~$i$ steps. Thus for $\eta_m$, one at most uses $i_m|V_{in}|+log_2(i_m)$ variables. 
			
			\item Similar estimation holds as an upperbound for the \s{release} part of specification type P2, but there is a need to add state variable of $\Theta_j$. Thus a conservative estimation is $j_m|V_{in}|+2\times log_2(j_m)$.
			
			\item A high-level control block uses at most two state variables, and has at most four ports. Each port is modeled as a variable. 
			
			\item All input ports of event monitor (for monitoring $\chi^i_{in}$) can be syntactically replaced by global input ports. All inputs of a OR-gate can be syntactically replaced by the output port of an event monitor or global output ports. The output of a OR-gate can be syntactically replaced by the input port of a high-level block. 
			
		\end{itemize}

	\end{list1}
	
	Therefore, the total number of variables is bounded by 
	
	\[\Omega(|V_{in}|+|V_{out}|+\sum _{m=1\ldots k} (2k|V_{out}|+(i_m|V_{in}|+log_2(i_m))+(j_m|V_{in}|+2\times log_2(j_m))+6)) \]
	
	The number of constraints created in each unroll can be understood by how data is flowed from source to destination, which follows the ordering (here we omit the logic gates): external input ports $\Rightarrow$  monitor controllers $\Rightarrow$ skeleton controllers $\Rightarrow$ resolution actors $\Rightarrow$ external output ports. Evaluating constraints in single cycle takes time linear to the number of actors, and evaluating  $\Omega = (i_1+1)+(i_2+1)+\ldots +(i_k+1)$ rounds takes polynomial time. Therefore, given an assignment over all variables, deciding whether the formula is violated is done in polynomial time.

	
\end{proof}

\noindent The following result shows that, the bound is still valid even without the above mentioned restriction. 

\begin{lemma}
	\normalfont
	For \GXW specification $\varrho \rightarrow \bigwedge_{m=1\ldots k} \eta_m$,   if the specification is not realizable, then there exists a counter-strategy which leads to violation of (1) or (2) in $\Omega$ steps, where $\Omega$ is bounded by the sum of  (i) the number of specifications $k$, and (ii)  the sum of all $i$ value defined within each $\phi^i_{in}$ of $\eta_m$.
	
\end{lemma}

\begin{proof}
	
	For general \GXW, the effect of enabling output variable to certain value within a given time point $\alpha$ is not carried to time point $\alpha+1$, as no \x is bundled with any output variable in the specification. 
	
	We again refer readers to Figure~\ref{fig:TimeLine}. Now, view the control of output variables being governed by an imaginary SAT solver. In each time point,  the produced output assignments should satisfy invariance properties of P5, while being constrained by the locking condition.	When possible, it tries to produce output assignment that turns $\rho^0_{out}$ to \s{true}, such that the output is no longer constrained to $\varrho_{out}$. For example in Figure~\ref{fig:TimeLine}, starting from time $i_1$, the SAT solver has independently $i_2+1$ opportunities to make the $\s{release}$ part to \s{true}. If it succeeds, then conflict does not appear. Otherwise, counter strategy is produced latest at time $i_1+1+i_2$.
	
\end{proof}

\begin{lemma}
	\normalfont
	For a \GXW specification, deciding whether the specification is realizable is in \s{PSPACE}.  
\end{lemma}

\begin{proof}
	
	We check if the specification is not realizable $\Sys$ by the following: non-deterministically provide input variable assignments for $(i_1+1)+(i_2+1)+\ldots +(i_k+1)$ times, and check for all output assignments for $(i_1+1)+(i_2+1)+\ldots +(i_k+1)$ rounds, it is possible to violate the specification. Checking whether it is possible to violate the specification can be done in \s{PSPACE}: The process is similar to the above unroll case, but for  each variable $v_{out}$, instead of setting $Res_{v_{out}}.\s{A}$ to $\s{true}$, 
	we simply let output variable to process
	$(i_1+1)+(i_2+1)+\ldots +(i_k+1)$ different copies, meaning that one can freely select the value in each round. The total memory used is $((i_1+1)+(i_2+1)+\ldots +(i_k+1)) |V_{out}|$, which is polynomial to the problem size. 
	
	Then given input assignment for $(i_1+1)+(i_2+1)+\ldots +(i_k+1)$ rounds, one can use the memory to check if for all possible output variable assignments of 
	they all unfortunately lead to conflict. If so, then report that the specification is not realizable. Therefore, deciding whether the specification is unrealizable is in \s{NPSPACE} (non-determinisstic input assignment + \s{PSPACE} complexity for checking if conflict appears).
	
	As \s{NPSPACE} = \s{DPSPACE} = \s{PSPACE}, and an LTL specification is either realizable or unrealizable,  deciding whether the specification is realizable is in \s{PSPACE}.  
\end{proof}

\vspace{2mm}
\noindent Notice that although the complexity for checking if a $\GXW$ specification is realizable is in \s{PSPACE}, the algorithm presented previously only sets every $Res_{v_{out}}.\s{A}$ as a constant that does not change over time. This creates a simpler structure for the implemented controller. Soundness is still guaranteed by performing an unroll to the above mentioned bound.

\begin{ProcessStep}
	\begin{small}
		\SetKwInOut{Input}{Input}
		\SetKwInOut{Output}{Output}
		\Input{LTL specification $\phi =  \varrho \rightarrow \bigwedge_{m=1\ldots k} \eta_m$, input variables $V_{in}$, output variables $V_{out}$, partial controller implementation $\Sys= (V_{in}, V_{out}, Act, \tau$) with unknown parameters, integer unroll bound $\Omega$ } 
		\Output{2QBF constraint $(V_{\exists}, V_{\forall}, \Upsilon)$, where $V_{\exists}$ and $V_{\forall}$ are sets of Boolean variables, and $\Upsilon$ is a quantifier free constraint over variables in $V_{\exists}\cup  V_{\forall}$
		}

		\textbf{let} $\Upsilon_{a}, \Upsilon_{g} := \s{true}$\;
		\textbf{let} $V_{\exists}, V_{\forall}$ := \textsf{NewEmptySet}();	
		
		\lForEach{$v_{out} \in V_{out}$}{
			$V_{\exists} := V_{\exists} \cup \{ Res_{v_{out}}.\s{A}  \}$ 
		} 
		

		\textbf{let} $\Xi$ be the evaluation ordering of $\Sys$ \;
		\For{$\alpha = 0 \ldots \Omega$}{
			
			\ForEach{$\xi \in \Xi$}{
				
				\If{$\xi \in \tau$}{
					\tcc{$\xi$ is w wire; encode biimplication among two ports }
					Let $\xi$ be $(\s{source} \dashrightarrow \s{dest})$\;
					$V_{\forall}.add(\s{source}_{\alpha})$\;
					$V_{\forall}.add(\s{dest}_{\alpha})$\;	
					$\Upsilon_{a} := \Upsilon_{a} \wedge (\s{source}_{\alpha} \leftrightarrow \s{dest}_{\alpha})$\;	
				} \Else{
				\tcc{$\xi$ is an actor; encode transition  using index $\alpha$ and $\alpha+1$}
				$V_{\forall}.add(\textsf{GetStateVariable}(\xi, \alpha))$\;	
				$\Upsilon_{a} := \Upsilon_{a} \wedge (\textsf{EncodeTransition}(\xi, \alpha) )$	
				
			}
			
		}

		$\Upsilon_{a} := \Upsilon_{a} \wedge \bigwedge \textsf{VariableReplace}(\varrho, \alpha)$\;

		\For{$v_{out} \in V_{out}$}{	
			$\Upsilon_{g} := \Upsilon_{g} \wedge (\not\exists i, j: (Res_{v_{out},\alpha}.\s{input}_{i}=\s{true}) \wedge (Res_{v_{out},\alpha}.\s{input}_{j} = \s{false}))$\;	
		}
		
		\ForEach{$\eta_m$, $m=1\ldots k$}{		
			$p$ := \textsf{DetectPattern}($\eta_m$)\;
			
			\lIf{$p \in \{P5\}$}{
				$\Upsilon_{g} := \Upsilon_{g} \wedge \textsf{VariableReplace}(\eta_m, \alpha)$
			}
		}

	}

	\tcc{Add initial condition, to achieve bounded unroll}
	\ForEach{$\xi \in \Xi$}{
		\If{$\xi \not\in \tau$}{
			\tcc{$\xi$ is an actor; encode initial state with index $0$}
			$\Upsilon_{a} := \Upsilon_{a} \wedge (\textsf{EncodeInitialState}(\xi, 0) )$	
		}
	}

	\textbf{return} ($V_{\exists}, V_{\forall},  \Upsilon_{a} \rightarrow \Upsilon_{g}$)

	\caption{Generating 2QBF constraints for bounded unroll}
	\label{step.bounded.unroll}
\end{small}
\end{ProcessStep}

\end{document}